\documentclass[11pt,letterpaper]{article}

\usepackage[utf8]{inputenc}
\usepackage{xspace, complexity}

\usepackage{fullpage}
\usepackage[colorlinks=true,linkcolor=blue,citecolor=blue,urlcolor=blue]{hyperref}
\usepackage{stmaryrd} 
\usepackage{subfigure}

\usepackage{multirow}
\usepackage[export]{adjustbox}

\usepackage{authblk}

\usepackage{amsmath,amsfonts,amssymb,amsthm}
\usepackage{latexsym}
\usepackage{enumitem}
\usepackage{verbatim} 

\usepackage{xcolor}

\usepackage{ulem} 
\newtheorem{theorem}{Theorem}
\newtheorem{lemma}[theorem]{Lemma}

\newtheorem{corollary}[theorem]{Corollary}
\newtheorem{claim}[theorem]{Claim}

\newtheorem{question}{Question}

\theoremstyle{remark} 
\newtheorem*{remark}{Remark}

\def\defn#1{\textit{\textbf{\boldmath #1}}}

\newcommand{\NN}{\ensuremath{\mathbb{N}}\xspace}

\newcommand{\RN}{\ensuremath{\mathbb{R}}\xspace}

\newcommand{\EQ}{\ensuremath{\exists\mathbb{Q}}\xspace}
\newcommand{\ER}{\ensuremath{\exists\mathbb{R}}\xspace}
\newcommand{\VR}{\ensuremath{\forall\mathbb{R}}\xspace}
\newcommand{\problemname}[1]{\textnormal{\textsc{#1}}\xspace}
\newcommand{\ETR}{\problemname{ETR}}                         
\newcommand{\ETRINV}{\problemname{ETR-INV}}                  
 
 \newcommand{\ISO}{\problemname{ISO}}                         
 \newcommand{\RAC}{\problemname{RAC}}                           

\newcommand{\Gapp}{\textit{Geogebra applet:}\xspace}

\newcommand{\hide}[1]{}

\newcommand{\remove}[1]{{}}

\title{Recognizing Penny and Marble Graphs is Hard \\ for Existential Theory of the Reals}

\author[1]{Anna Lubiw\thanks{alubiw@uwaterloo.ca}}
\author[2]{Marcus Schaefer
\thanks{mschaefe@depaul.edu}}
\affil[1]{University of Waterloo}
\affil[2]{DePaul University}

\begin{document}

\maketitle

\begin{abstract}
We show that the recognition problem for penny graphs 
(contact graphs of unit disks in the plane) is \ER-complete, that is, computationally as hard as the existential theory of the reals, even if a combinatorial plane embedding of the graph is given.
The exact complexity of the penny graph recognition problem has been a long-standing open problem. 

We lift the penny graph result to three dimensions and show 
that the recognition 
problem for marble graphs (contact graphs of unit balls in three dimensions) is \ER-complete. 

Finally, we show that rigidity of penny graphs is \VR-complete and look at grid embeddings of penny graphs 
that are trees. 
\end{abstract}

\section{Introduction}
\label{sec:intro} 

\defn{Penny 
graphs} are contact graphs of unit
disks in the plane, or equivalently, contact graphs of disks of equal radius in the plane.
For example, $K_4-e$ is a penny graph as the realization in Figure~\ref{fig:K4me} shows, but $K_4$ is not.
\begin{figure}[htb]
    \centering
\includegraphics[height=1in]{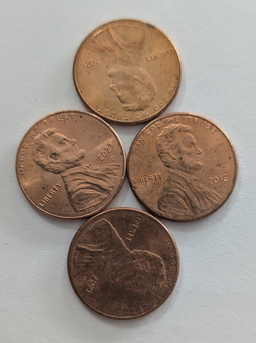}
    \caption{A penny graph realization of $K_4-e$.
    }
    \label{fig:K4me}
\end{figure}

Similarly, Figure~\ref{fig:grid} proves that the $4 \times 4$ grid is a penny graph, but as the middle illustration shows, the $4 \times 4$ grid has many degrees of freedom; this flexibility can be reduced by bracing the sides of the grid, as shown in the right illustration; the braced $4\times 4$ grid behaves like a rhombus (a property we will exploit later).

\begin{figure}
    \centering
\includegraphics[height=1.5in]{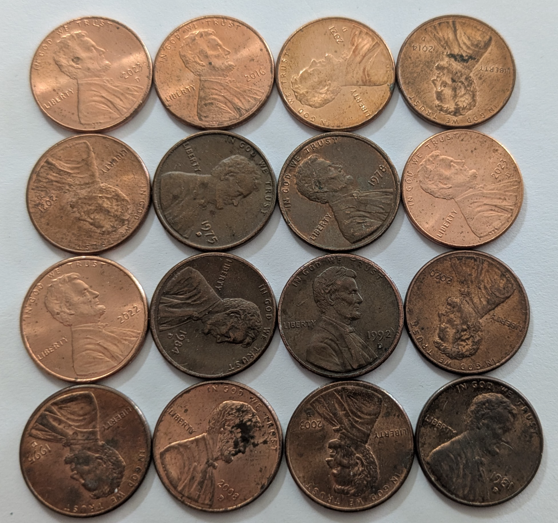}
 \hskip .1in
\includegraphics[height=1.5in]{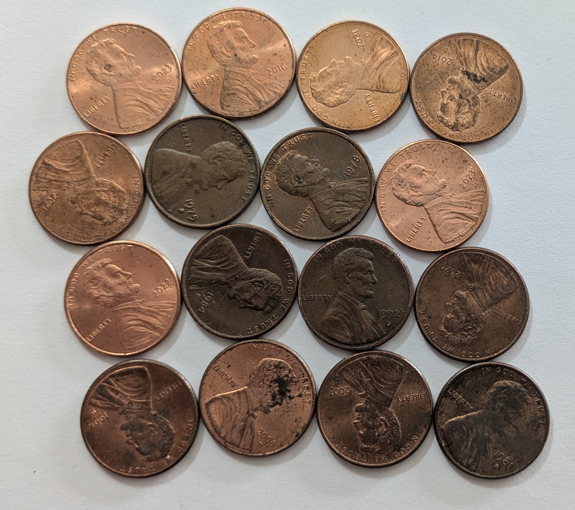}
 \hskip .1in
\includegraphics[height=1.5in]{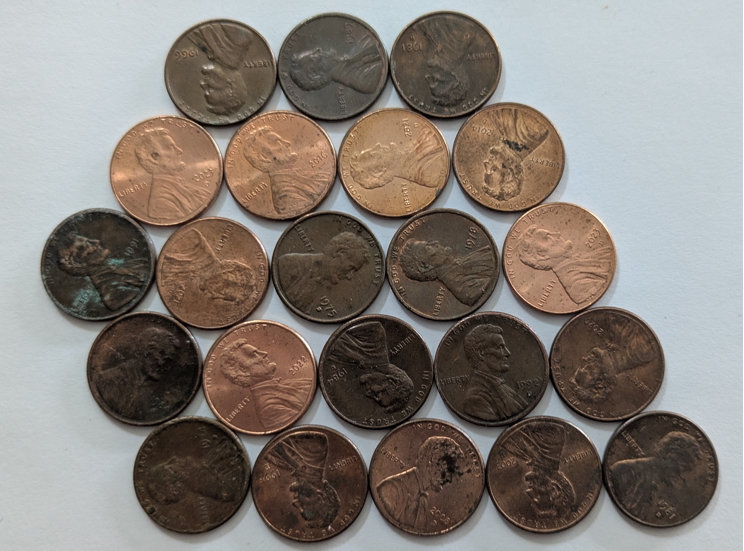}
    \caption{{\em (Left)} the $4 \times 4$ grid as a penny graph; {\em (middle)} a perturbed realization of the $4\times 4$ grid; {\em (right)} a braced $4\times 4$ grid with one degree of freedom. 
    }
    \label{fig:grid}
\end{figure}

Our main result is that the recognition problem for penny graphs 
is complete for \ER, the existential theory of the reals.  
Until now, the best complexity result for penny graphs was due to Breu and Kirkpatrick~\cite{BK96,BK98} who showed 
in 1996
that recognizing penny graphs is \NP-hard, and remains hard even if the radii of the disks can vary slightly.   As noted by them, and
also by Hlin{\v e}n\'y and Kratochv\'il~\cite{HK01},
membership in NP is not at all obvious---the issue, for the straight-forward approach,   
is how to represent the coordinates of the disk centers.
The \NP-hardness proof of Breu and Kirkpatrick encoded the true/false settings of Boolean variables as different choices in the combinatorial embedding.  This approach left open 
the complexity of the problem 
when the input includes a specified combinatorial embedding. 
Our proof shows that this variant is \ER-complete as well.  

The \ER-hardness of penny graph recognition has the usual consequences: the coordinates in a realization of a penny graph can be algebraic numbers of arbitrary (algebraic) complexity.\footnote{A number is {\em algebraic} if it is the root of an integer polynomial; we can take the size of the encoding of such a polynomial as a measure of the complexity of the algebraic number.} We show that for trees this is not the case; rational coordinates with double-exponential precision are sufficient to realize any penny graph 
that is a tree.

Our proof for penny graphs allows us to prove and/or strengthen \ER-completeness results for recognition problems for several other graph classes including matchstick and unit-distance graphs and, in $\RN^3$, marble graphs and contact graphs of balls (of possibly differing radii). 

Another consequence of our reduction is that testing rigidity of a given penny graph configuration is complete for \VR, the {\em universal theory of the reals}. Previously, \VR-completeness of rigidity testing had been known for linkages (Schaefer~\cite{S13}) and matchstick graphs (Abel, Demaine, Demaine, Eisenstat, Lynch, and Schardl~\cite{ADDELS25}).

A \defn{matchstick graph} is a graph that has a planar straight-line drawing with all edges of unit length.  The class of penny graphs is properly contained in the class of matchstick graphs. See Figure~\ref{fig:unit-dist}, which also shows classes of unit-distance graphs. 
A \defn{strict unit-distance graph} is a graph whose vertices can be placed at distinct points in the plane such that two vertices are distance~1 apart  if and only if there is an edge between them. 
A \defn{weak unit-distance graph} is a subgraph of a unit-distance graph, i.e., the condition is ``if'' rather than ``if and only if''.\footnote{Some nomenclature omits ``strict'' and some omits ``weak''. We use both to avoid confusion.}  
A matchstick graph can equivalently be defined as a \defn{non-crossing} weak unit-distance graph, 
meaning that it has a weak unit-distance representation in which edges do not 
intersect 
except at a common endpoint. 
Figure~\ref{fig:unit-dist} also shows the class of non-crossing strict unit-distance graphs which lies between penny and matchstick graphs.

\begin{figure}
    \centering
    \includegraphics[width=0.6\linewidth]{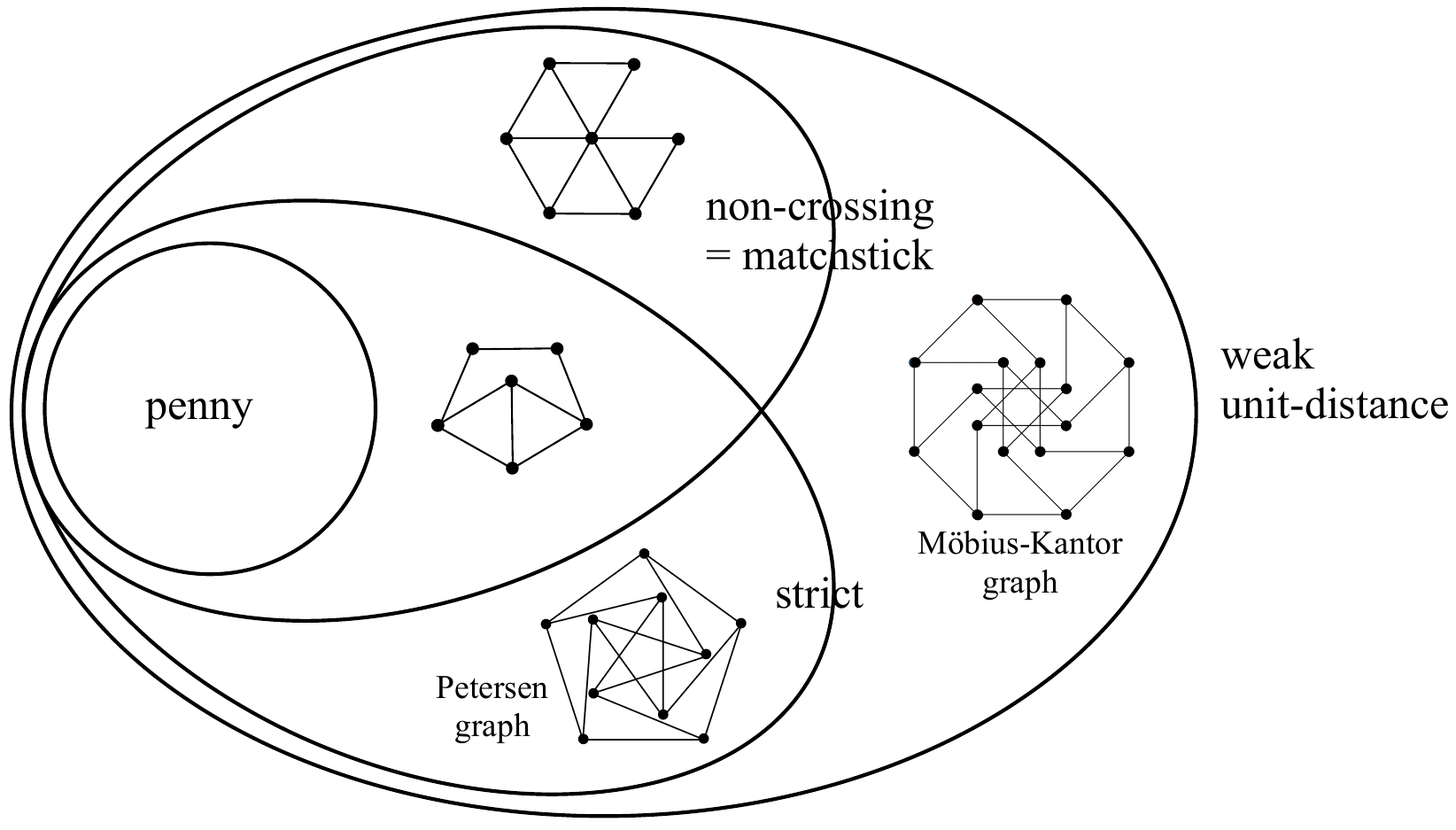}
    \caption{Subclasses of unit-distance graphs with examples to show that containments are proper. 
    }
    \label{fig:unit-dist}
\end{figure}


Schaefer~\cite{S13}
showed that 
weak/strict unit-distance graphs are \ER-complete to recognize~\cite{S13}, but his hardness construction 
requires crossings.
Abel, Demaine, Demaine, Eisenstat, Lynch, and Schardl~\cite{ADDELS25} proved that the problem 
for weak unit-distance remains \ER-complete for crossing-free configurations, in other words, the recognition problem for matchstick graphs is \ER-complete even if a combinatorial plane embedding of the graph is specified; our main result yields a new proof of this theorem.
We also prove \ER-completeness of recognizing 
non-crossing strict unit-distance graphs, with or without a  specified combinatorial embedding.
Thus all classes shown in Figure~\ref{fig:unit-dist} are \ER-complete to recognize.

Penny graphs generalize to \defn{marble graphs} which are contact graphs of equal radius balls in~$\RN^3$. 
Hlin\v en\'y~\cite{H97} proved that the problem of recognizing marble graphs is 
\NP-hard. 
We prove that the problem is \ER-complete. 
We also show \ER-completeness of recognizing contact graphs of balls of possibly different radii in $\RN^3$.  This problem is easy in two dimensions, since the contact graphs of disks
(the \defn{coin graphs})
are precisely the planar graphs~\cite{K36,S03}.

\paragraph{Main Result.}
Our hardness results are proved via a reduction from an \ER-complete problem called \ETRINV which asks for real values for variables in the range $[\frac{1}{2},2]$ to satisfy a formula~$\varphi$ that is a conjunction of constraints of the form $x + y = z$ and $x \cdot y = 1$. 
For more details on \ETRINV, see Section~\ref{sec:ETR} or the survey~\cite{SCM26}.
Our reduction produces a graph $G$ with a combinatorial plane embedding. A \defn{combinatorial plane embedding} is a collection of circular permutations of 
the edges incident to each vertex 
that corresponds to a plane embedding of the graph~\cite[Chapter 4]{MT01}.
The reduction has the following properties. 

\begin{theorem}\label{thm:penniesER}
There is a polynomial-time algorithm that takes as input an \ETRINV\ formula $\varphi$ and constructs a graph $G$ and a combinatorial plane embedding $D$ of $G$ such that
\begin{itemize}
    \item 
    if $\varphi$ is true, then $G$ has a penny graph embedding realizing $D$,
    and
    \item if $\varphi$ is false, then $G$ is not a  weak unit-distance graph.
\end{itemize}
\end{theorem}

The conditions in the theorem imply that the constructed graph $G$ has a penny graph realization if and only if $\varphi$ is true.  Thus Theorem~\ref{thm:penniesER} implies that recognizing penny graphs---with or without a specified combinatorial plane embedding---is \ER-hard. 
Since in the case that $\varphi$ is false, $G$ cannot even be realized as a 
weak unit-distance graph, we can draw a stronger conclusion.

\begin{corollary}\label{cor:penniesER}
 It is \ER-complete to recognize penny graphs, non-crossing strict unit-distance graphs,  matchstick graphs, 
 and weak/strict unit-distance graphs with or without a specified combinatorial embedding.
\end{corollary}

For matchstick and weak/strict unit-distance graphs the results were known and we provide a different proof.

\paragraph{Paper Outline.}
Theorem~\ref{thm:penniesER} is proved in Sections~\ref{sec:linkages} and~\ref{sec:penny-graphs}.  
In Section~\ref{sec:linkages} we 
construct from an ETR-INV formula~$\varphi$ 
a linkage---a graph with specified integer edge lengths and angle constraints between some pairs of incident edges---that has a plane
realization iff $\varphi$ is true. 
In Section~\ref{sec:penny-graphs} we show that the particular linkages we construct can be simulated by penny graphs.
The \ER-completeness results for marble and ball contact graphs in $\RN^3$ are proved in Section~\ref{sec:marbles}.
Section~\ref{sec:R} shows $\VR$-completeness of testing rigidity of penny graphs, and Section~\ref{sec:GE} looks at grid embeddings of penny graphs 
that are trees.

\paragraph{Gadgets.}
We have implemented most of the gadgets used in this paper as geogebra applets which can be explored interactively~\cite{Geo26}. These gadgets are available through the arxiv version of the paper~\cite{LS25}. We will point out the names of the applets in the corresponding figures.\\

We conclude the current section by discussing history and background, and describing \ETRINV\ in more detail.

\subsection{History and Background}\label{sec:HB}

Penny graphs are the special case of coin graphs in which all coins have the same radius; this creates a connection with Koebe's theorem, which states that all planar graphs are coin graphs~\cite{K36}. As far as we can tell, the term ``penny graph'' was first introduced in the context of coin graphs and Koebe's theorem in the 1990s by Hartsfield, Ringel, and Breu~\cite{HR94,B96}. Penny graphs, under different names, are older though. 

Erd\H{o}s~\cite{E46} introduced the minimum distance graphs 
and proved that they are planar (in fact, they are the penny graphs), and therefore have at most $3n-6$ edges; this led to research---still very active---on edge bounds in penny graphs; Harborth~\cite{H74} showed a sharp upper bound on the number of edges in penny graphs of $\lfloor 3n-\sqrt{12n-3}\rfloor$. 

Penny graphs have also been defined as 
contact graphs of unit (or equal-radius) disks in the context of {\em intersection graphs}. Breu and Kirkpatrick~\cite{BK98} proved that recognizing intersection graphs of unit disks, the so-called \defn{unit-disk graphs}, is \NP-hard.  They mention that their result extends to penny graphs (which they call ``unit disk touching graphs''). 
They make this result explicit in a conference paper~\cite{BK96} and strengthen it to show that the problems of recognizing intersection or contact graphs of disks is \NP-hard even if the radii of the disks are in an interval $[1,\rho]$ for any fixed $\rho \ge 1$.
Hlin{\v e}n\'y and Kratochv\'il~\cite{HK01} later showed that recognizing \defn{disk graphs}, intersection graphs of disks, is \NP-hard as well.

Breu and Kirkpatrick raise the question of membership in \NP: ``[a]lthough grid-constrained versions of our recognition problems
(like grid graph recognition [...]) are in fact \NP-complete, it is
not clear that the unconstrained versions are in \NP''. 
They recognize the relationship of their problems to the theory of the reals, writing that ``membership in \PSPACE\ follows directly from results of Canny~\cite{C88}''. 
This result is more formally stated by Hlin{\v e}n\'y and Kratochv\'il in~\cite[Proposition 4.4]{HK01}. In this 2001 paper, the authors also write ``An important question is whether the recognition of disk graphs belongs to \NP. The
obvious approach to guess a disk representation and check it, does not work immediately,
since the coordinates and radii of disks in the representation may not be expressible
using polynomial number of bits.''

The \ER-completeness of both disk graphs and unit disk graphs was eventually shown by McDiarmid and M\"{u}ller~\cite{McDM13}, who also gave bounds on the required precision of a realization~\cite{McDM10}. 
Bowen, Durocher, L{\"o}ffler, Rounds, Schulz, and T{\'o}th~\cite{BDLRST15} showed that penny graph recognition remains \NP-hard for trees with a given embedding. It is also known that testing whether a graph is a subgraph of a penny graph is \NP-hard for embedded caterpillars~\cite{CCN20} and for trees without a given embedding~\cite{C20}. The complexity of checking whether a tree without embedding is a penny graph is open, even \NP-hardness is not known in this case, though there is an \NP-hardness result for outerplanar graphs without embedding~\cite{BLNN21}.
Most of the \NP-hardness results are based on  
the logic engine of 
Eades and Whitesides~\cite{EW96},
which does not lend itself to \ER-hardness proofs.

Penny graphs generalized to higher dimensions give us unit-ball contact graphs, known as marble graphs in dimension $3$; here we touch on one of the most famous problems in geometry, Kepler's conjecture, now known as Hales' theorem, on the density of optimal sphere packings. Unit-ball contact graphs are known to be \NP-hard to recognize in $\RN^d$ for $d = 3, 4$ by Hlin\v en\'y~\cite{H97} and for $d = 8$ and $d = 24$ by Hlin\v en\'y and Kratochv\'il~\cite{HK01}. The latter paper is also the most recent survey on intersection and contact presentations with disks and balls. Both papers conjecture that the problem is \NP-hard for arbitrary dimension. 

Matchstick graphs, or equivalently, non-crossing weak unit distance graphs, seem to trace back to Harborth's ``Match Sticks in the Plane''~\cite{GW94}, in which he shows that there is a $3$-regular matchstick graph on $8$ vertices and a $4$-regular matchstick graph on $52$ vertices, now known as the {\em Harborth graph}.  Harborth conjectured that the upper bound of  $\lfloor 3n-\sqrt{12n-3}\rfloor$ edges for penny graphs extends to matchstick graphs.
The conjecture was settled only recently (in the affirmative) by Lavoll\'{e}e and Swanepoel~\cite{LS24}.

We already mentioned that the \ER-completeness of 
recognizing matchstick graphs was shown by Abel, Demaine, Demaine, Eisenstat, Lynch, and Schardl~\cite{ADDELS25}. Earlier \NP-hardness results are due to Eades and Wormald~\cite{EW90} and Whitesides~\cite{W92}. Cabello, Demaine and Rote~\cite{CDR07} showed that matchstick graphs are \NP-hard to recognize even if the graph is $3$-connected (and infinitesimally rigid).  

Matchstick graphs can also be seen as a special case of graphs with fixed-length embeddings, where each edge is assigned a length that has to be realized. This notion, with and without the planarity constraint, has also been investigated, starting with Yemini~\cite{Y79}, and with linkages in place of graphs (allowing edges and vertices to overlap). The \ER-completeness results from~\cite{S13} belong in this context. \ER-hardness here was foreshadowed by a result of 
Maehara~\cite{M91} who showed that any algebraic number can occur as a distance between two vertices in a rigid unit-distance graph. 

\subsection{The Existential Theory of the Reals and \ER-completeness}\label{sec:ETR}

How do we even determine whether a given graph is a penny graph? The problem is not obviously decidable, as it involves irrational numbers: any representation of $K_3$ as a penny graph requires at least one irrational coordinate (since a triangle with unit sides has irrational height). The answer, maybe surprisingly, comes from logic: We can express that $K_3$ is a penny graph as
\begin{equation}\label{eq:K3penny}
\begin{split}
(\exists x_1,y_1,x_2,y_2,x_3,y_3)\ & (x_1-x_2)^2+(y_1-y_2)^2 = 1\ \wedge \\
 & (x_1-x_3)^2+(y_1-y_3)^2 = 1\ \wedge \\
 & (x_2-x_3)^2+(y_2-y_3)^2 = 1. 
\end{split}
\end{equation}
Here $(x_i,y_i)$, $i \in \{1,2,3\}$, are the three centers of the pennies which we, arbitrarily but conveniently, modeled as having diameter $1$, so that the centers have distance $1$. 

Sentence \eqref{eq:K3penny} is true if interpreted over the real numbers, so $K_3$ is a penny graph. More generally, an $n$-vertex graph $G = (V,E)$ is a penny graph if and only if the sentence
\begin{equation}\label{eq:Gpenny}
\begin{split}
(\exists x_1, y_1, \ldots, x_n, y_n)\ & \bigwedge_{v_iv_j \in E} (x_i-x_j)^2+(y_i-y_j)^2 = 1\ \wedge \\
 & \bigwedge_{v_iv_j \notin E} (x_i-x_j)^2+(y_i-y_j)^2 > 1
\end{split}
\end{equation}
is true over the real numbers. Sentences~\eqref{eq:K3penny} and~\eqref{eq:Gpenny} are written in a language that logicians
refer to as the existential theory of the reals (\ETR). \ETR\ is defined as the set of all existentially quantified Boolean formulas
using constants $\{0,1\}$, arithmetical operators $\{+,\cdot\}$ and comparison operators $\{=, <\}$, and which are true sentences (no free variables) when interpreted over the real numbers. Examples include, $(\exists x)\ x\cdot x = 1+1$, and  $(\exists x_1, \ldots, x_n)\ x_1 = 1 + 1 \wedge \bigwedge_{i=1}^{n-1} x_i\cdot x_i = x_{i+1}$, but not 
$(\exists x)\ x\cdot x +1 = 0$, because it is not true over the real numbers, or $(\exists x)\ x = y$, because it is not a sentence (there is a free variable, $y$).

Tarski~\cite{T48} introduced the method of quantifier elimination and used it to show that \ETR\ is decidable. The decidability of \ETR\ then implies that the penny graph problem is decidable, since we saw how to translate any penny graph recognition problem efficiently into the language of \ETR. Our main result in this paper is that there is a translation in the reverse direction: we can efficiently (in polynomial time) translate a question of the form $\varphi \in \ETR$ into a penny graph recognition problem. So \ETR\ and penny graph recognition have the same computational complexity.

This bi-translatability of computational problems to \ETR\ was first noticed in the early 1990s and eventually led to the definition of a complexity class, \ER, the \defn{existential theory of the reals}, as the set of all decision problems that reduce (translate) in polynomial time to \ETR. It is known that $\NP \subseteq \ER \subseteq \PSPACE$~\cite{S91, C88}; in particular, \ER-hardness implies \NP-hardness. For history and more detailed background, see~\cite{M26,SCM26}. 
The class \ER\ captures the complexity of many computational problems involving real numbers, many of them in computational geometry, such as the segment intersection graphs~\cite{KM94}, the rectilinear crossing number~\cite{B91}, the art gallery problem~\cite{AAM22} and many others. The recent compendium~\cite{SCM26} lists more than 200 problems complete for \ER, with about half of them being classified as coming from computational geometry.

\paragraph{\ETRINV.} Our \ER-hardness reduction will use a restricted form of \ETR\ introduced by 
Abrahamsen, Adamaszek, and Miltzow~\cite{AAM22} as part of their proof that the art gallery problem is \ER-complete. This restricted form is known as \ETRINV; it restricts the range of the variables to the interval $[\frac{1}{2},2]$ and allows only two real operations: addition and inversion, that is, we have \defn{constraints} of the form $x = y+z$ and $x\cdot y = 1$.
Abrahamsen and Miltzow~\cite{AM19} in a separate note worked out the full details of \ER-completeness of \ETRINV, that is, they showed how to reduce \ETR\ to \ETRINV. This allows us to work with \ETRINV\ as our starting problem.

\paragraph{\VR.} Just like \NP\ has \coNP, \ER\ has \VR, the \defn{universal theory of the reals}; a problem belongs to \VR\ if its complement belongs to \ER. While we could avoid \VR, by stating that the negation of a problem is in \ER, it is often more natural to work with \VR.

\section{From \ETRINV\ to Linkages with Angle Constraints}
\label{sec:linkages}

In this section we show how to translate questions of the form $\varphi \in \ETRINV$ into a linkage realizability problem. We introduce our linkage model in Section~\ref{sec:linkage_model}. At the heart of the translation is a historic linkage, {\em Hart's A-frame}, which we describe in Section~\ref{sec:Aframe}. Hart's A-frame allows us to build what we call a flex gadget, introduced in Section~\ref{sec:flexgadget}. With this flex gadget we can construct the more sophisticated gadgets we need for the reduction, see Sections~\ref{sec:overview} and ~\ref{sec:construction-details}. The main properties of the reduction are stated in Section~\ref{sec:linkages-wrapup}.

\subsection{Angle-Constrained Linkages}
\label{sec:linkage_model}

For our purposes, a \defn{linkage} is a graph with positive 
edge-weights. 
The edges of a linkage are referred to as \defn{bars}.
A \defn{configuration} or \defn{realization} of a linkage assigns 
points in $\mathbb{R}^2$ 
to the vertices so that the resulting straight line edges have the specified lengths. 
A configuration is \defn{non-crossing} if edges do not intersect except at a common endpoint; otherwise we say that the configuration is \defn{crossing} (note that this
includes the case
of two vertices assigned to the same point and the case of a vertex assigned to a point in a non-incident embedded edge).

An \defn{angle-constraint} specifies the angle between an ordered pair of incident edges, and 
an \defn{angle-constrained} 
linkage consists of a linkage together with some angle constraints.
A configuration of an angle-constrained linkage must satisfy the angle constraints.
In a realization, 
the  two edges of an angle constraint
need not be consecutive at their common vertex.

We reduce \ETRINV to the question of whether an angle-constrained linkage has a realization; see Theorem~\ref{thm:ETRINVlinkage} below. 
To prove our main result (Theorem~\ref{thm:penniesER}) for penny graphs we show that the particular linkages we construct can be modeled as penny graphs, i.e., from an angle-constrained linkage we construct a penny graph preserving non-crossing realizability.
The construction relies on special properties of the linkages.
In particular, our linkages have
integer edge lengths, small degree vertices, and each angle constraint limits some face angle in the combinatorial embedding to some multiple of $30^\circ$.

\subsection{Hart's A-frame}
\label{sec:Aframe}
A bar of a linkage has fixed length so if one endpoint is stationary, the other endpoint moves in a circle. 
Our reduction makes use of a linkage that converts circular motion to straight-line motion.  This allows us to model a variable as the distance between two moving parallel bars of the linkage.

The Peaucellier linkage is the most well-known linkage for straight-line motion, but we find it more convenient to have the fixed points of the linkage in the outer face of the embedding. We therefore use Hart's A-frame, 
aka Hart's second inversor,
as shown in Figure~\ref{fig:A-frame}, with horizontal base $AD$ and with edge lengths as shown in the figure.
In Hart's A-frame, 
$AC$ and $DG$ are single bars with internal vertices $B$ and $E$, respectively. 
This is captured in our angle-constrained linkage model by replacing $AC$ by two
bars $AB$ and $BC$ 
with the angle between them
at $B$ constrained to $180^\circ$, and similarly for $DE$ and $EG$.  

\begin{figure}[htb]
    \centering   \includegraphics[width=.3\textwidth]{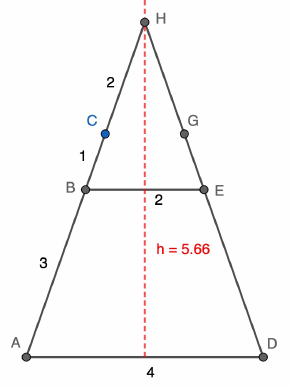}
    \hskip .1in
    \includegraphics[width=.3\textwidth]{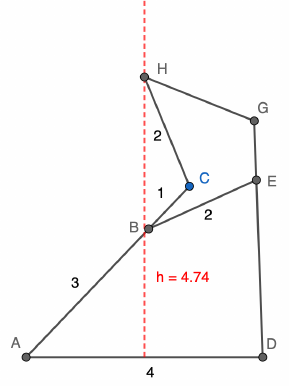}
    \hskip .1in
    \includegraphics[width=.3\textwidth]{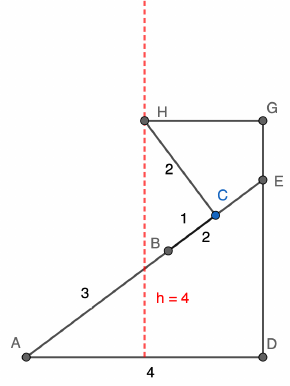}
   
    \caption{
    Hart's A-frame. $AC$ and $DG$ are single bars of length 4.  {\em (left)} When $C$ is collinear with $A$ and $H$, the apex $H$ is at height 
    $h = 4 \sqrt 2 \approx 5.66$.
    {\em (middle)} As $C$ moves in a circle centered at $A$, $H$ moves vertically downward and $h$ decreases.  {\em (right)} In the degenerate configuration when $C$ becomes collinear with $A$ and $E$, bar $DG$ is vertical and bar $HG$ is horizontal, so the limit for 
    non-crossing configurations is $h=4$. \Gapp Hart's-A-frame.ggb; vertex $C$ can be moved.}
    \label{fig:A-frame}
\end{figure}

\begin{claim}
\label{claim:A-frame}
In Hart's A-frame, 
as $C$ moves on a circle centered at $A$, the apex $H$ moves on a vertical line bisecting $AD$. For edge lengths as specified in Figure~\ref{fig:A-frame}, 
the height~$h$ 
takes on all values in the interval 
$[-4\sqrt{2}, 4\sqrt{2} \,] \approx [-5.66, 5.66]$.\footnote{Throughout, we round to two decimal places for presentation purposes.}
The configuration 
is crossing iff $h \in [-4,4]$.
Every value of $h$ in the open interval 
$(4,4\sqrt{2})$
corresponds to exactly two configurations of the A-frame, a left-leaning one and a right-leaning one, both 
non-crossing.
Furthermore, for right-leaning [or left-leaning] configurations, $h$ and all the angles of the A-frame are  continuous monotone functions of the angle~$CAD$.
\end{claim}

The first part of this claim is given as an exercise in several sources, e.g.,~\cite{HowRound}.  We include a proof of Claim~\ref{claim:A-frame} in  Appendix~\ref{app:AFrame}.

\subsection{The Flex Gadget}\label{sec:flexgadget}

Using Hart's A-frame we 
construct a \defn{flex gadget} as follows.  See the contents of the light blue rectangle in Figure~\ref{fig:flex}. 
Place two A-frames scaled by factor 2 side-by-side with a gap of length 2 between them at the base. Join their apexes $H$ and $H'$ with a bar of length 10.

\begin{figure}[htb]
    \centering   \includegraphics[width=.48\textwidth]{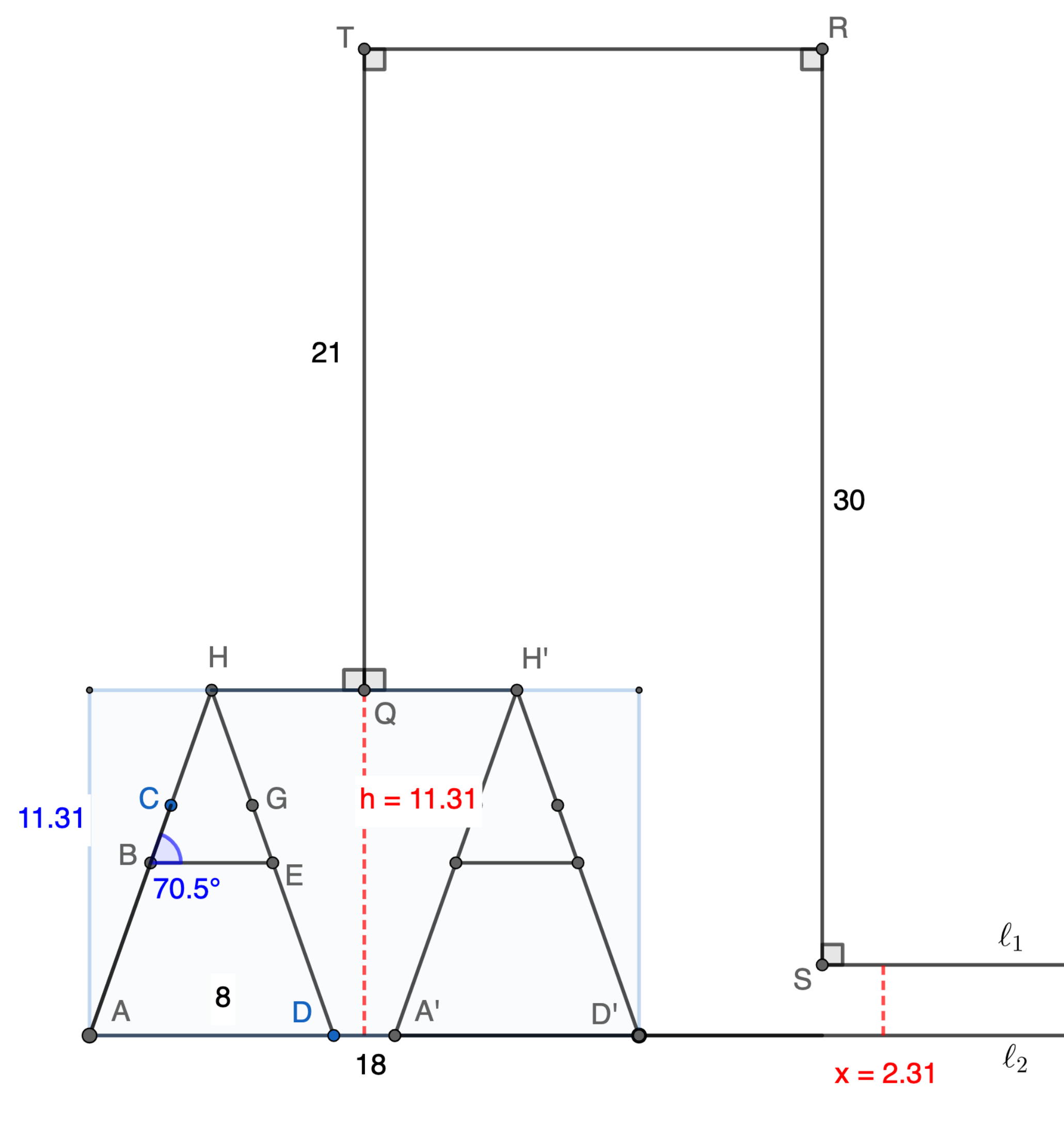}
\hskip .1in
\includegraphics[width=.48\textwidth]{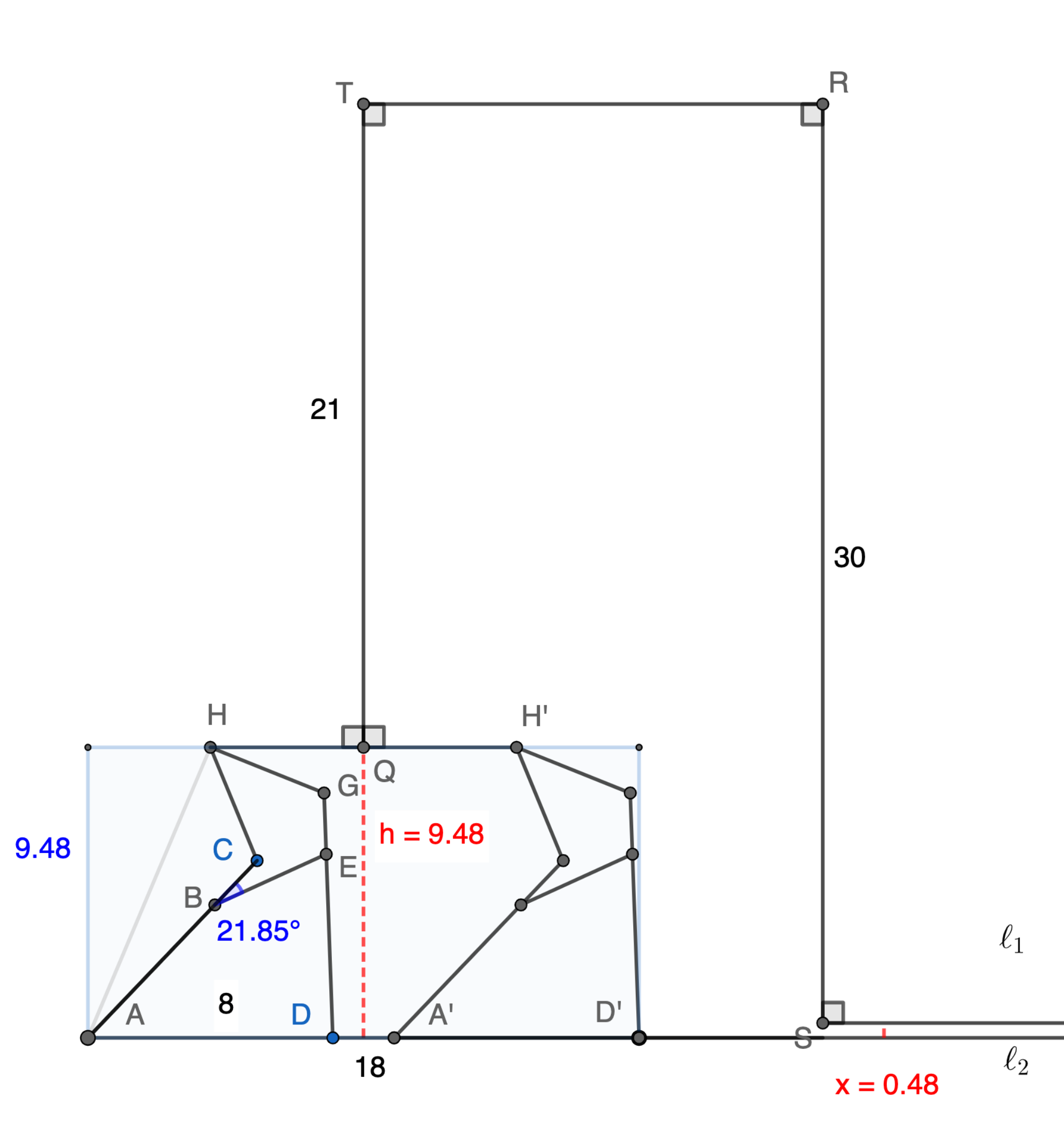}
    \caption{A flex gadget (in the light blue rectangle) determines a channel formed by two parallel bars, $\ell_1$ and $\ell_2$, separated by distance $x$ in a range that includes $[\frac{1}{2},2]$. 
    Drawing conventions used throughout: 
    bars and angles drawn in black are of fixed size; distances and angles drawn in blue or red are not fixed.
    \Gapp flex-gadget.ggb; vertex $C$ can be moved.
}
    \label{fig:flex}
\end{figure}

\begin{claim}
\label{claim:flex}
In a flex gadget the bar $HH'$ joining the apexes remains parallel to the bar $AD'$ of the bases.
A flex gadget has width $18$.
Configurations with the apex bar above the base bar are non-crossing iff $h \in $ 
$(8, 8 \sqrt{2}\, ] \approx (8,11.31]$.
\end{claim}
\begin{proof}
Refer to Figure~\ref{fig:flex}.
By Claim~\ref{claim:A-frame}, $H$ and $H'$ move on vertical lines.  Thus,
if $H$ and $H'$ were at different heights, bar $HH'$ would be longer than 10.  The two A-frames are then forced to act synchronously and $HH'$ remains parallel to the base.
The width of the flex gadget is $18$ by construction. 
The statement about height follows from Claim~\ref{claim:A-frame}.
\end{proof}

In Section~\ref{sec:channelrestrictgadget}, we show how to
limit the height of a flex gadget to the range $[9.5,11]$.
This has the important implication (via Claim~\ref{claim:flex}) that the flex gadget is restricted to non-crossing configurations.

To model a variable $x$, 
we add bars
with angles restricted to $90^\circ$
as shown in Figure~\ref{fig:flex} to create two parallel bars called a \defn{channel} (see $\ell_1, \ell_2$ in the figure) with distance in a range that includes $[\frac{1}{2},2]$.  
By limiting the height of the flex gadget to the range $[9.5,11]$
we limit the channel distance to $[\frac{1}{2},2]$.

In our construction we will often use the larger channel of distance $x + 30$, and narrow to a channel of distance $x$ only when necessary. 
It may seem 
arbitrary that we place the top of the larger channel (at point $T$ in the figure) $21$ units above the top of the flex gadget,  
but we need room to
turn a channel from horizontal to vertical as shown later~on, and a difference of $30$ between the larger and narrower channels makes it easy to read off the value of $x$.

\subsection{Overview of the Construction}
\label{sec:overview}
Having represented each variable $x$ of the formula $\varphi$ as a horizontal channel of height $x + 30$, the overall construction is as shown in Figure~\ref{fig:overall}. 
There is a horizontal channel 
and a vertical channel 
for each variable.
From now on we refer to the distance of a channel as its \defn{width}, regardless of whether 
the channel is horizontal or vertical.

The channel widths vary depending on the values of the variables but the gaps between the channels have fixed lengths, though not all the same. 
There are addition and inversion gadgets (drawn here as rectangles) corresponding to the constraints
of~$\varphi$.
For each occurrence of a variable in a constraint we add an extra 
horizontal channel connected to the constraint's gadget.
  
When two channels cross, the width of each channel must be maintained.  If the channels correspond to different variables, the crossing must be ``free'' in the sense that the widths of the horizontal channel and the vertical channel are independent.  
However, in order to transmit the chosen value for a variable to its gadgets, when two channels corresponding to the same variable cross each other, we enforce a one-to-one relationship between their widths.  Enforcing the same width seems tricky.  Instead, we ensure that all the horizontal channels corresponding to variable $x$ have width $x + 30$, and all the vertical channels corresponding to  $x$ have width $\sqrt{3}(x+30)$.   

\begin{figure}[htb]
    \centering   \includegraphics[width=.4\textwidth]{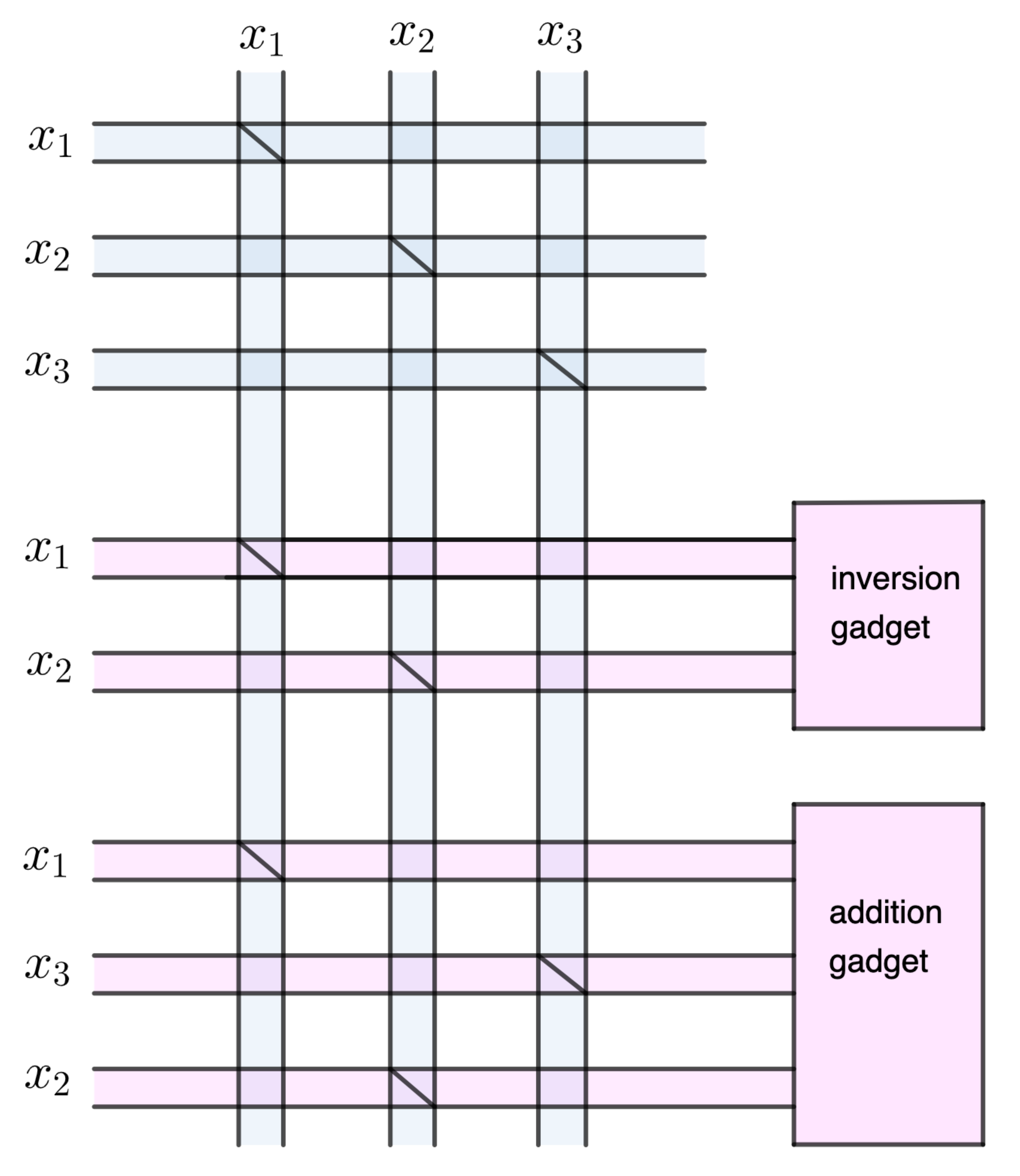}
    \caption{Overall plan for an instance of \ETRINV with three variables $x_1, x_2, x_3$ and constraints $x_1 \cdot x_2 = 1$ and  $x_1 + x_3 = x_2$.
    A diagonal segment where two channels cross indicates a relationship between their widths.}
    \label{fig:overall}
\end{figure}

In the following subsection 
we construct the following gadgets:

\begin{itemize}
\item The \defn{channel-restricting gadget} restricts the 
height of a flex gadget to the interval $[9.5,11]$ 
which restricts the flex gadget to non-crossing configurations and restricts the corresponding variable to $[\frac{1}{2},2]$.
Channel-restricting gadgets are placed at the beginning of the primary horizontal channel representing each variable (the blue horizontal channels in Figure~\ref{fig:overall}).  They limit all subsequent gadgets to non-crossing configurations. 
\item The \defn{turn gadget} enforces the relationship between a horizontal channel of width $w$ and a vertical channel of width $\sqrt{3} w$.  Turn gadgets are used in the cross-over gadgets and the addition and inversion gadgets.
\item \defn{Cross-over gadgets} are used whenever a horizontal channel crosses a vertical channel.  There are two versions depending on whether the two channels correspond to the same variable or not.  
\item The \defn{addition gadget} enforces a constraint $x_i + x_j = x_k$.
\item The \defn{inversion gadget} enforces a constraint $x_i \cdot x_j = 1$.  (To avoid subscripts, we will write this as $x \cdot y = 1$.)
\end{itemize}




\subsection{Details of the construction}
\label{sec:construction-details}

\subsubsection{Channel-restricting gadget and non-crossing properties.}
\label{sec:channelrestrictgadget}
We restrict the width of a large horizontal channel representing a variable $x$ to the interval 
$[30.5, 32]$ (i.e., restricting $x$ to  $[\frac{1}{2},2]$)
by adding bars as shown in Figure~\ref{fig:exact-interval}(middle)
to restrict the height of the flex gadget to $[9.5,11]$. 
The bars are added to the primary horizontal channel representing each variable, i.e., the blue horizontal channels shown in Figure~\ref{fig:overall}.   
This has the important 
consequence of  restricting the flex gadget to non-crossing configurations
by Claim~\ref{claim:flex}. 
(Nothing constrains the  channel-restricting gadget itself to non-crossing configurations.)

\begin{figure}[htb]
    \centering   \includegraphics[height=2.6in]{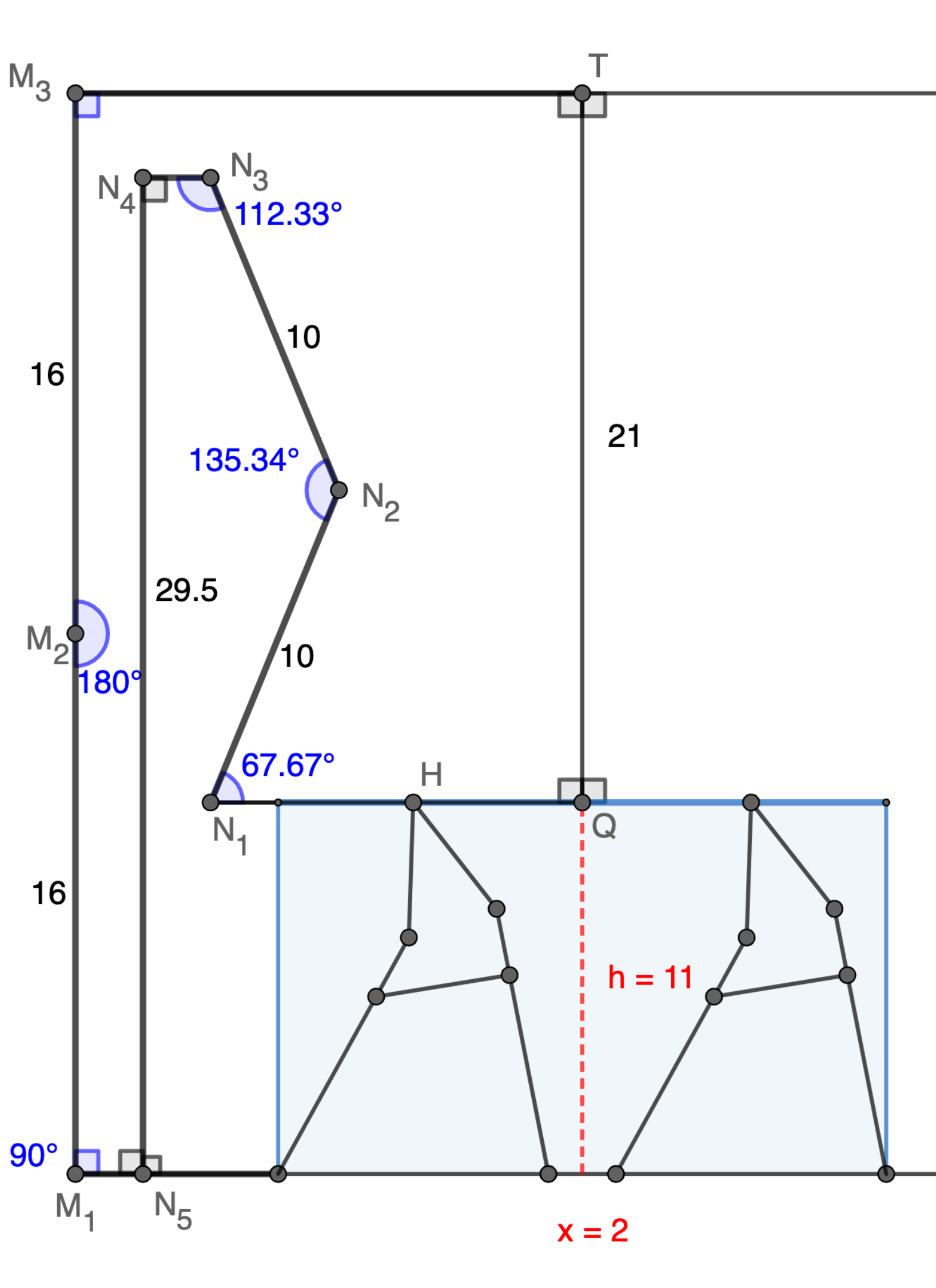}
    \includegraphics[height=2.6in]{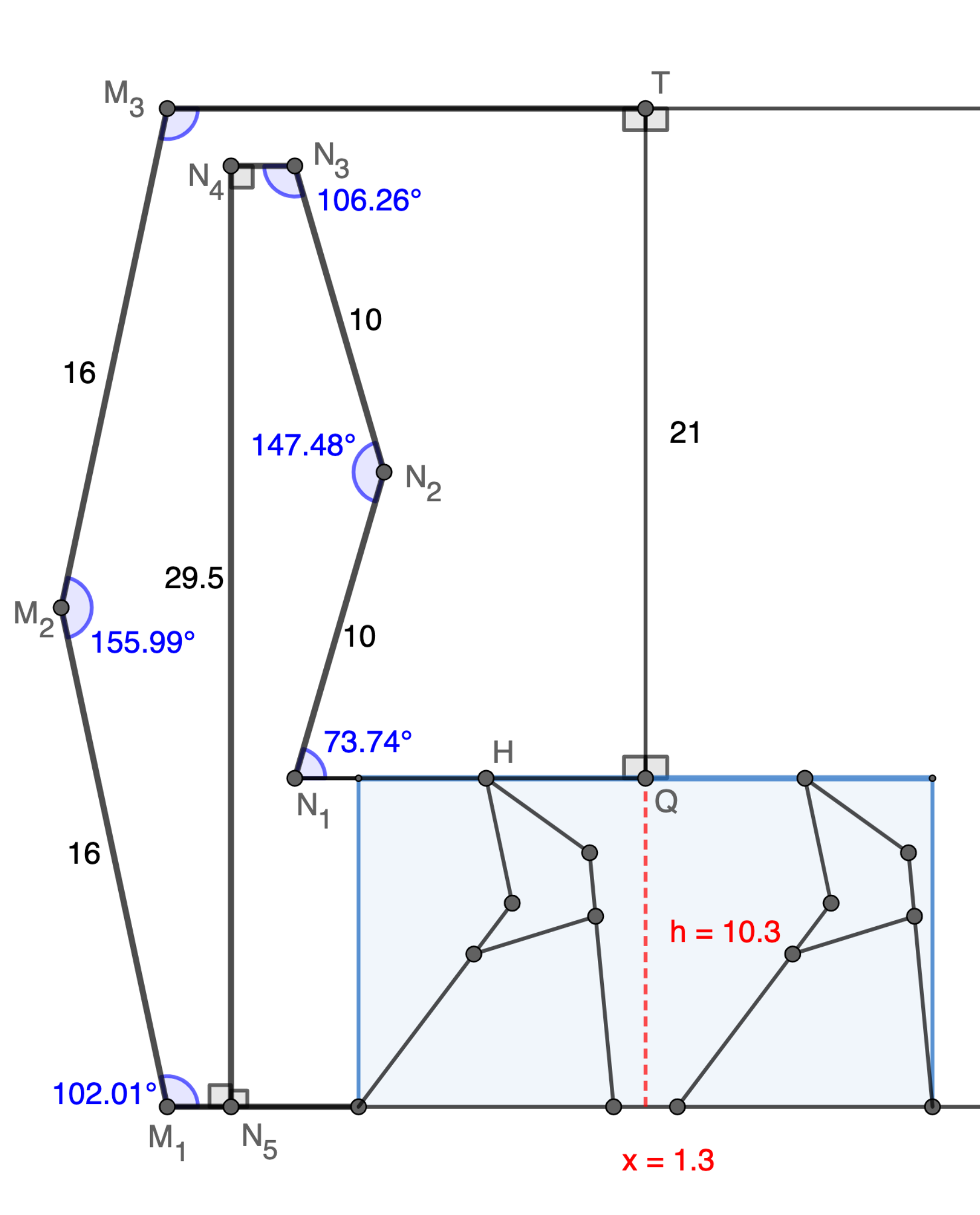}
    \includegraphics[height=2.6in]{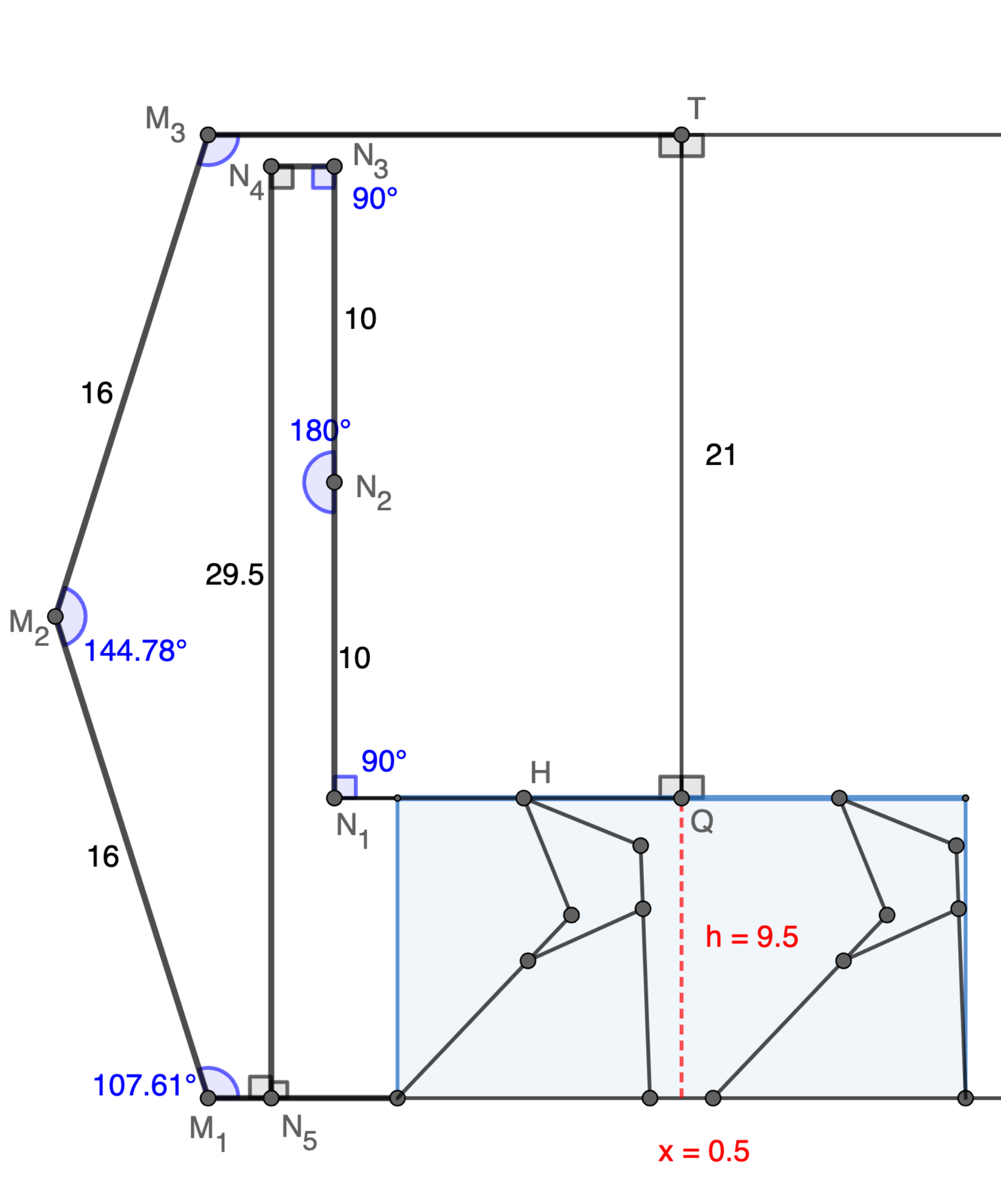}

    \caption{
Adding extra 
bars (drawn thick) at nodes $M_1,M_2,M_3$ and $N_1, \ldots, N_5$ to the left of the flex gadget (in the light blue rectangle) to constrain 
the height of the flex gadget to $[9.5,11]$.
Fixed angles are drawn in black; flexible angles in blue.
\Gapp flex-gadget-range-constrained.ggb; a slider controls the value of $x$.
}
    \label{fig:exact-interval}
\end{figure}

\begin{claim}
\label{claim:channel-restricting}
The channel-restricting gadget in Figure~\ref{fig:exact-interval} 
allows a flex gadget of height $h$ iff $h \in [9.5,11]$. Furthermore, all these values of $h$ can be realized via non-crossing configurations of the channel-restricting gadget.
\end{claim}
\begin{proof}
Points $M_1$ and $M_3$ are vertically aligned and the 
distance between them is at most $32$ because of the  
two bars $M_1 M_2$ and $M_2 M_3$.  Thus the maximum height of the flex gadget is $11$ (see  Figure~\ref{fig:exact-interval}(left)), corresponding to an $x$ value of $2$.
Similarly, the minimum $x$ value of $0.5$, corresponding to the minimum height of the flex gadget at $9.5$, is enforced by bars $N_1 N_2$ and $N_2 N_3$ (see Figure~\ref{fig:exact-interval}(right))). 
For the second statement in the claim, note that the configurations in the figure are non-crossing and intermediate configurations are non-crossing by monotonicity of the flex gadget.
\end{proof}

\subsubsection{Turn gadget}\label{sec:turngadget} 
As mentioned above, for a variable $x$, the
plan is that the corresponding horizontal channels have width 
$x+30$ and the corresponding vertical channels have width $\sqrt{3}(x+30)$.  
This relationship between a horizontal and vertical channel is enforced by the turn gadget which is used as a building block in  the crossing gadgets and the addition and inversion gadgets.
The turn gadget is shown in Figure~\ref{fig:crossing-gadget}(left). It is entered from the left by a horizontal channel of width $x+30$.
The width of the vertical channel is controlled 
by double-size flex gadgets
(labelled $F_2$ and $F_3$) that have base length $2 \cdot 16 + 2 = 34$, and the height 
of non-crossing configurations 
in the interval $(16, 16\sqrt{2}\,] \approx (16,22.63]$.

\begin{figure}[htb]
    \centering  \includegraphics[width=.39\textwidth]{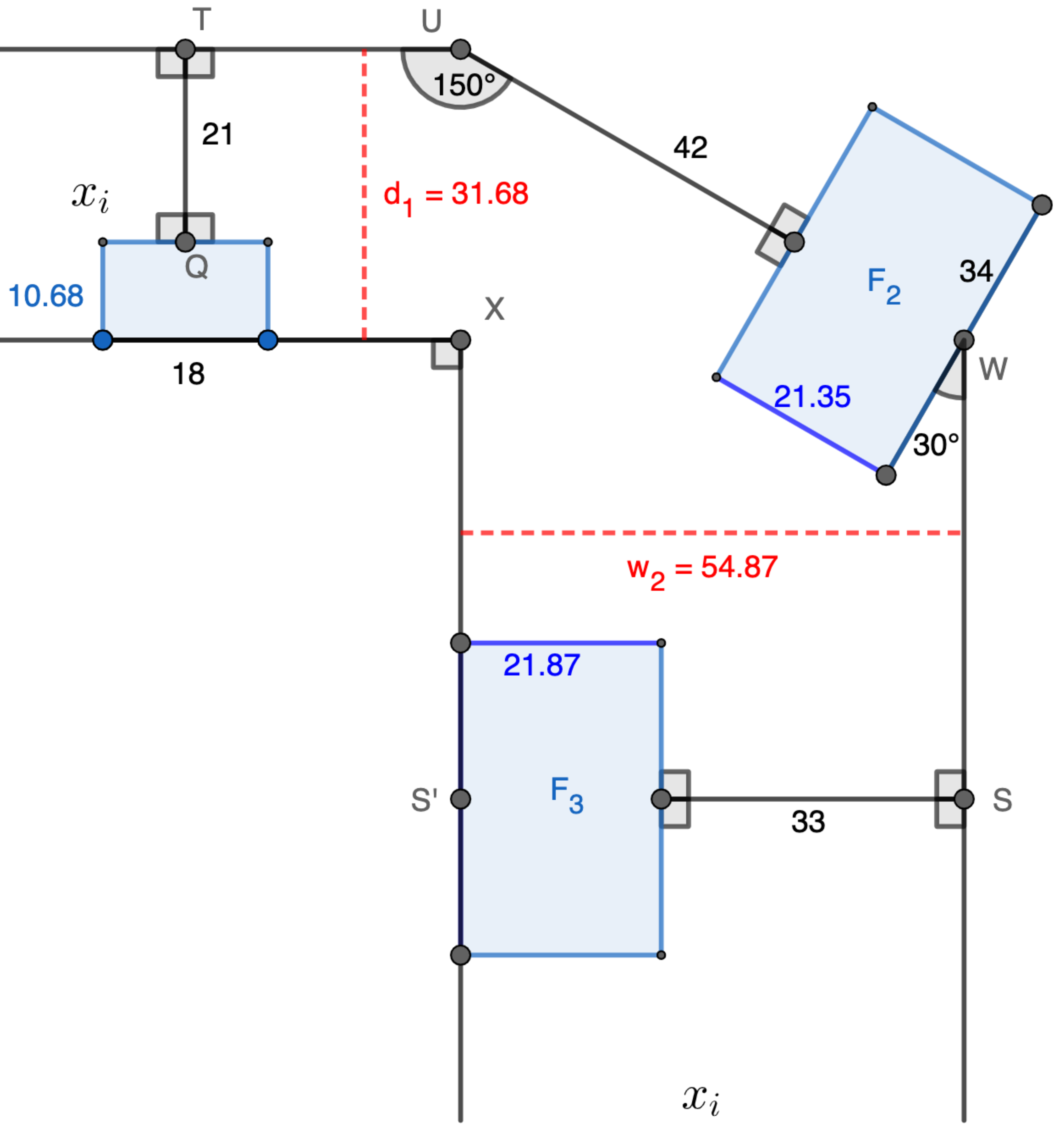}
        \hskip0.2in \includegraphics[width=.45\textwidth]{figures/crossing.pdf}

    \caption{
    {\em (left)} The turn gadget enforces  $w_2 = \sqrt{3} d_1$.  
    {\em (right)} The cross-over gadget for $x_i,x_j$ with $i \ne j$ enforces $d_1=d_2$ and $w_1 = w_2$.    \Gapp cross-over-gadget-different-variables.ggb; sliders control the values of $x_i$ and $x_j$.
 }
    \label{fig:crossing-gadget}
\end{figure}

\begin{claim} 
\label{claim:turn-gadget}
The turn gadget in Figure~\ref{fig:crossing-gadget}(left) works correctly:
if 
the horizontal channel has width $d_1 \in [30.5,32]$, then the vertical channel has width $w_2 =  \sqrt{3}d_1$.  Furthermore, these configurations are non-crossing.
\end{claim}
\begin{proof}
Points $U$ and $X$ are vertically aligned.  
The flex gadget at $S'$ forces $W$ to be horizontally aligned with $X$.  
Therefore, triangle $UXW$ is a $60^\circ,90^\circ,30^\circ$ right-angle triangle with side lengths $d_1, \sqrt{3} d_1, 2d_1$.  Because the $F_2$ flex gadget at $W$ is double size, the diagonal $UW$ can accommodate any choice for $d_1$ in its range. 
We next show that the distance $XW$ accommodates $x_i$ in the necessary range. 
Recall that $d_1 = x_i + 30$. 
For $x_i= \frac{1}{2}$, $d_1 = 30.5$, so we need $XW = \sqrt{3}d_1 \approx 52.83$, and we need flex gadget $F_3$ to have height $\sqrt{3}d_1 - 33 \approx 19.83$, which is in the  range of its non-crossing configurations.
For $x_i=2$, $d_1 = 32$, so we need $XW \approx 55.43$, and we need $F_3$ of height $22.43$, also in the range of its non-crossing configurations.
By continuity (Claim~\ref{claim:A-frame}), all values of $x_i$ in the range $[\frac{1}{2},2]$ are achievable via non-crossing configurations.    
\end{proof}

\subsubsection{Cross-over gadgets}\label{sec:crossgadgets}  
%
%
There are two versions of the cross-over gadget, one for when the corresponding variables are different, and one for when they are the same. 
The first version, for $x_i, x_j$, $i \ne j$ is shown in 
Figure~\ref{fig:crossing-gadget}(right).  It uses a standard flex gadget for the horizontal channel and a double-size flex gadget for the vertical channel.

\begin{claim}
\label{claim:cross-over-gadget}
The cross-over gadget for $x_i, x_j, i \ne j$ in Figure~\ref{fig:crossing-gadget}(right)  works correctly: if $d_1$ and $w_1$ lie in  $[30.5,32]$, then $d_2=d_1$ and $w_1=w_2$.  Furthermore, these configurations are non-crossing.  
\end{claim}
\begin{proof} 
Points $U$ and $X$ are vertically aligned.  
The flex gadget and bar between $S'$ and $S$ force $W$ to be horizontally aligned with $X$.  Similarly, the flex gadget and bars between $P'$ and $P$ force $V$ to be horizontally aligned with $U$.  Thus $d_2 = d_1$.  A similar argument shows $w_1 = w_2$.  All values of $x_i$ and $x_j$ in the range $[\frac{1}{2},2]$ are realizable (appealing to Claim~\ref{claim:turn-gadget} for the vertical channel).
\end{proof}

\begin{figure}[htb]
    \centering
    \includegraphics[width=.5\textwidth]{figures/crossing-same.pdf}
    \caption{The cross-over gadget for $x_i,x_i$ enforces 
    $d_1=d_2$, $w_1 = w_2$, and 
    $w_1 = \sqrt{3} d_1$.  
    \Gapp cross-over-gadget-same-variable.ggb; a slider controls the value of $x_i$.
    }
    \label{fig:crossing-same-gadget}
\end{figure}

The version of the cross-over gadget for $x_i,x_i$ incorporates the turn gadget into the cross-over gadget as shown in Figure~\ref{fig:crossing-same-gadget}.

\begin{claim}
The cross-over gadget for $x_i, x_i$ in Figure~\ref{fig:crossing-same-gadget}  works correctly:
if $d_1 \in [30.5,32]$, then 
$d_1 = d_2$, $w_1 = w_2$, and $w_1 = \sqrt{3}d_1$.  Furthermore, these configurations are non-crossing. 
\end{claim}
\begin{proof}
By Claims~\ref{claim:turn-gadget} and~\ref{claim:cross-over-gadget}.
\end{proof}

\subsubsection{Addition gadget}\label{sec:addgadget}  
Consider the constraint $x_i + x_j = x_k$.  
From the overall plan (see Figure~\ref{fig:overall}), 
the gadget corresponding to this constraint is entered from the left by
horizontal channels of widths $x_i + 30$, $x_j + 30$ and $x_k + 30$.  We separate these channels by fixed distances of $90$ units.  The addition gadget is constructed as shown in Figure~\ref{fig:addition-gadget} using three double-size flex gadgets.  
\begin{claim}
The addition gadget in Figure~\ref{fig:addition-gadget} works correctly:
if $x_i, x_j \in [\frac{1}{2},2]$ and the corresponding incoming channels have width $x_i + 30$ and $x_j + 30$, respectively, then the channel for $x_k$ has width $x_i + x_j + 30$.  Furthermore, these configurations are non-crossing.  
\end{claim}
\begin{proof}
We prove that the gadget  enforces $x_i + x_j = x_k$, and does not otherwise restrict the values of the variables in the interval $[\frac{1}{2},2]$.

Because the distance between the incoming  
channels for $x_i$ and for $x_j$ is fixed at 90 units, 
we have $|GM| = (x_i +30) + (x_j +30) + 90$, so $|LM| = |GM|-120 = x_i + x_j +30$. 
By Claim~\ref{claim:turn-gadget} the two turns force $|LM| = |UV|$. 
Thus $x_i + x_j + 30 = x_k + 30$. 
Also by Claim~\ref{claim:turn-gadget} the gadget realizes (via non-crossing configurations) any values of $x_i, x_j, x_k$ in $[\frac{1}{2},2]$ that satisfy $x_i + x_j = x_k$. 
\end{proof}

\begin{figure}
    \centering
    \includegraphics[width=.4
\textwidth]{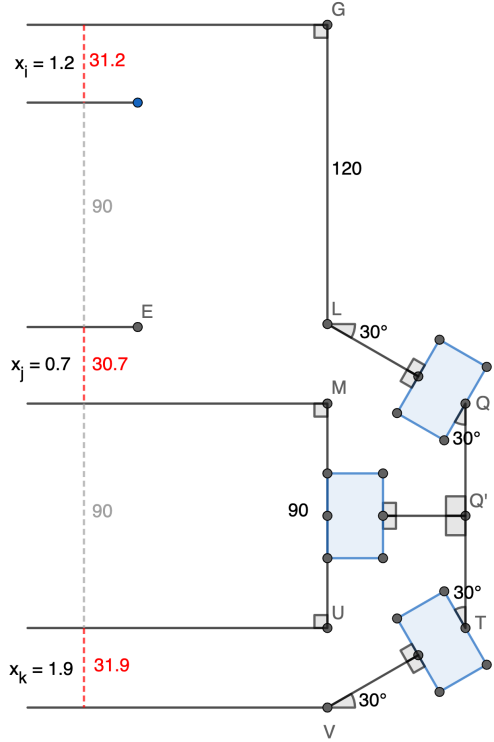}
    \caption{
    The addition gadget for $x_i + x_j = x_k$, illustrated here for $x_i=1.2$, $x_j=0.7$, and $x_k = 1.9$. 
    Each of the three incoming horizontal channels has width $x+30$ for value $x$.
    Note that $E$ and $L$ are not aligned.
        \Gapp addition-gadget.ggb; sliders control the values of $x_i$ and $x_j$.
    }
    \label{fig:addition-gadget}
\end{figure}

\subsubsection{Inversion gadget}\label{sec:inversiongadget} 
Consider the constraint $x \cdot y = 1$.  The corresponding inversion gadget is entered from the left by channels of widths $x + 30$ and $y +30$.  We separate the channels by fixed distance $110$.  
We turn the channel for $y$ to an upwards vertical channel.
Although the addition gadget worked with the wide channels, the inversion gadget cannot, so we reduce the channels to widths $x$ and $\sqrt{3} y$, respectively, as shown in Figure~\ref{fig:inversion-setup}.  
The basic idea of the inversion gadget is shown in Figure~\ref{fig:inversion-gadget}.
It requires us to simulate two collinear segments (shown in red) with flexible lengths. The gadget to accomplish that is shown in Figure~\ref{fig:inverting-detail} using two flex gadgets  scaled by one-and-a-half with width $2\times 12 + 2$ and height of non-crossing configurations in the range $(12,12\sqrt{2}\,] \approx (12, 16.97 ]$.

\begin{figure}[tb]
    \centering
    \includegraphics[width=0.4\linewidth]{figures/inverting-4-setup.pdf}
    \caption{    Set-up for the inversion gadget. Channels for variables $x$ and $y$, of width $x+30$ and $y+30$, respectively, enter the inversion gadget from the left. The top channel is narrowed to its actual value $x$.  The bottom channel is turned upward and then narrowed to width $\sqrt{3} y$. Contents of the dashed box are shown in subsequent figures. 
     \Gapp inversion-setup.ggb; a slider controls the value of $x$.
    }
    \label{fig:inversion-setup}
\end{figure}

\begin{figure}[htb]
  \centering
  \includegraphics[height=2.85in]{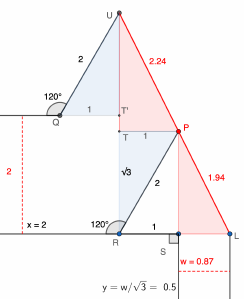}
\hskip 0.09in 
\includegraphics[height=2.85in]{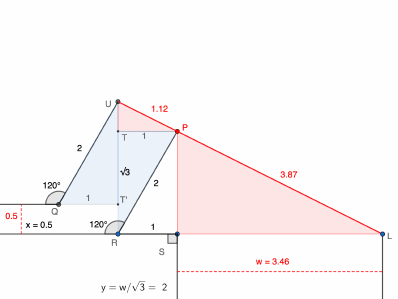}
  \caption{
    An inversion gadget forcing 
    $x \cdot y = 1$
    by means of the two similar red triangles.  Only black segments are bars of the linkage.  The two collinear red segments incident to $P$ must have adjustable lengths, and are implemented as shown in Figure~\ref{fig:inverting-detail}.    
    }
    \label{fig:inversion-gadget}
\end{figure}

\begin{figure}[htb]
    \centering
    \includegraphics[width=.49\textwidth]{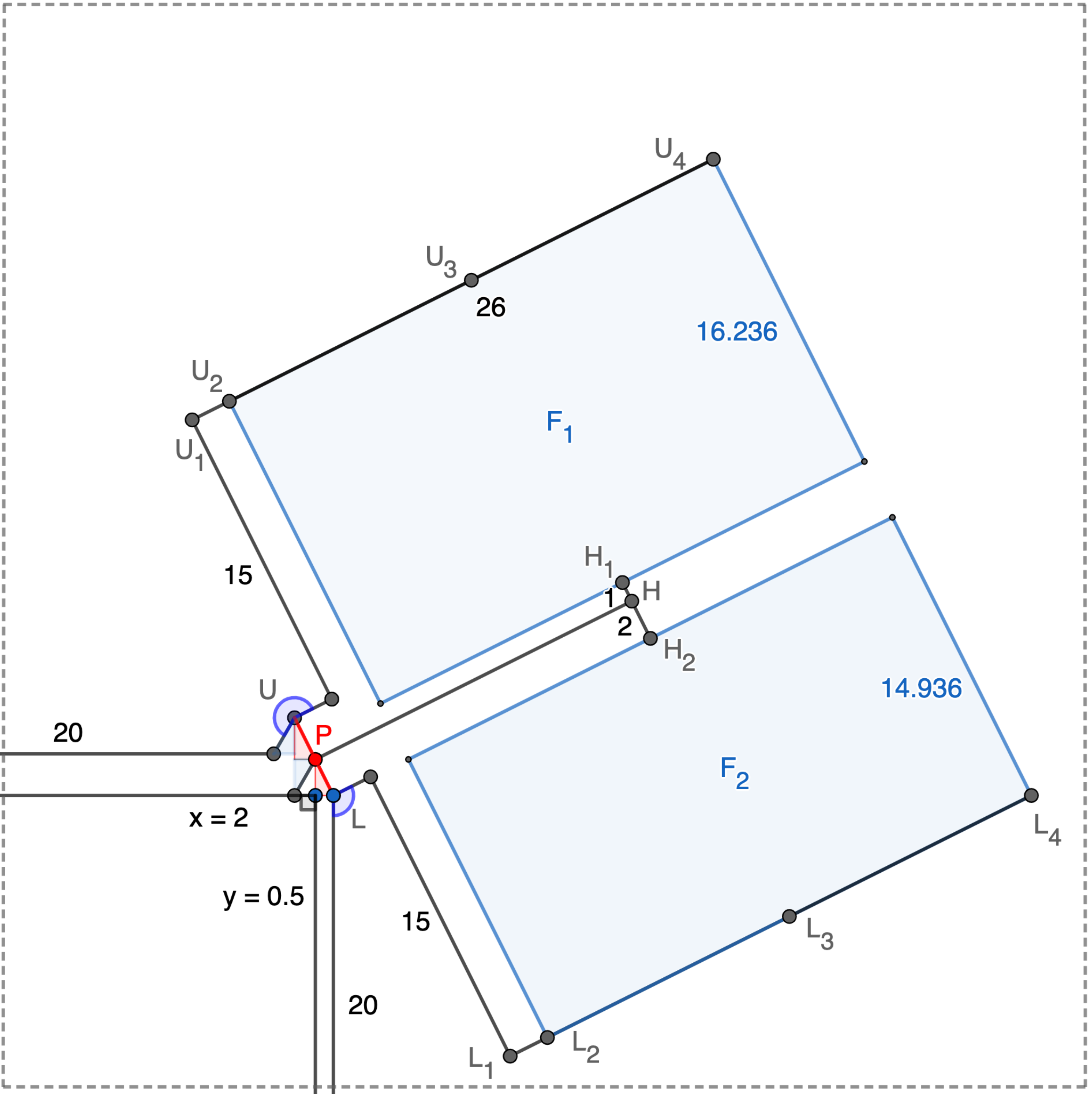}
 \includegraphics[width=.49\textwidth]{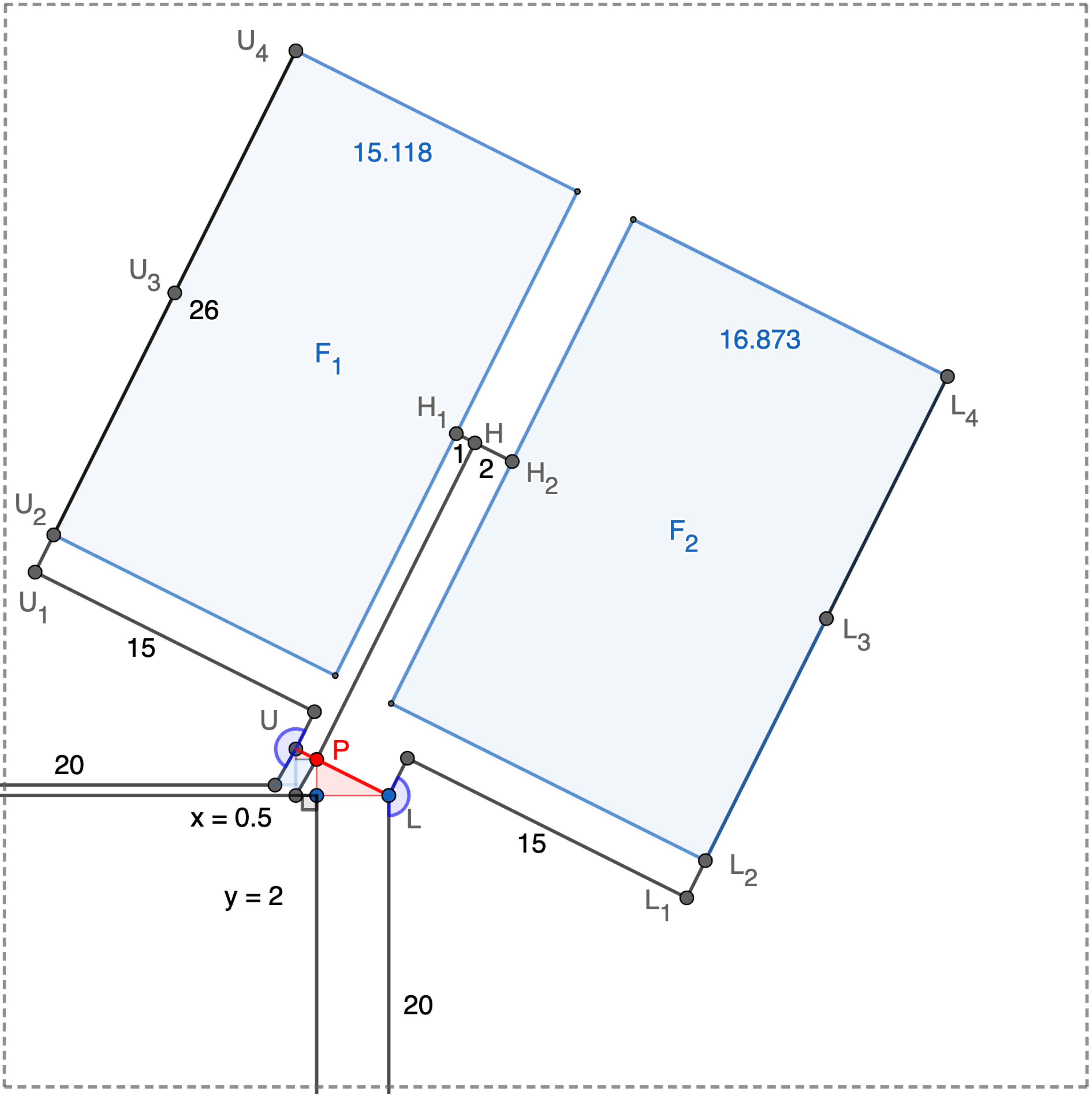}
    \caption{
   Detail of inversion gadget using flex gadgets $F_1$ and $F_2$ scaled by one-and-a-half with width $2\times 12 + 2 = 26$ and height of non-crossing configurations in the range $(12,12\sqrt{2}\,] \approx (12, 16.97 ]$.
   Right-angle symbols have been omitted (for example at $U_1, L_1, H$) but should be apparent from the two views.
    \Gapp inversion-detail.ggb; a slider controls the value of $x$.
    }
    \label{fig:inverting-detail}
\end{figure}

\begin{claim}
\label{claim:inversion}
The inversion gadget shown in Figures~\ref{fig:inversion-setup},~\ref{fig:inversion-gadget},~\ref{fig:inverting-detail} works correctly:
if $x \in [\frac{1}{2},2]$ and the corresponding incoming channel has width $x + 30$,  then the channel for $y$ has width $1/x + 30$ (i.e., $x \cdot y = 1$).  Furthermore, these configurations are non-crossing. 
\end{claim}
\begin{proof}
The set-up for the inversion gadget, as shown in Figure~\ref{fig:inversion-setup}, uses the turn gadget from Figure~\ref{fig:crossing-gadget}(left).
The resulting vertical channel has width $\sqrt{3}(y + 30)$. 
We claim that bar $EF$ narrows the vertical channel to width $\sqrt{3}y$.  This is because the right-angle triangle $EFG$ (where $G$ is not part of the construction) has side-lengths $60,30,30\sqrt{3}$.

Next, consider the construction in Figure~\ref{fig:inversion-gadget}. The bottom bar of the channel for $x$ and the left bar of the channel for $y$ meet at point $S$.  From points $R$ and $Q$ (which lie on vertical lines $1$ and~$2$ units, respectively, to the left of $S$) we construct bars of length~$2$ at a fixed angle of $120^\circ$ from horizontal.  
This creates three right-angle triangles of side-lengths $2,1,\sqrt{3}$, namely
$PRS$, $RPT$, and $UQT'$.

Consider the triangles $UTP$ and $PSL$ (shaded red).
Triangle $UTP$ has altitude $x$ and base~$1$; triangle $PSL$ has altitude $\sqrt 3$ and base $w = \sqrt{3} y$.
If the segments $UP$ and $PL$ are 
collinear, then the triangles $UTP$ and $PSL$ are similar so
$x / 1 = \sqrt{3} / \sqrt{3} y$, i.e., $x \cdot y = 1$.

Figure~\ref{fig:inverting-detail} shows how we force $UP$ and $PL$ to be collinear.  It remains to show that the flex gadgets $F_1$ and $F_2$ allow any values of $x$ and $y$ in the interval $[\frac{1}{2},2]$ that satisfy $x \cdot y = 1$. 
Recall that $F_1$ and $F_2$ are flex gadget with  height of non-crossing configurations in the (rounded) range $(12, 16.97]$).
Since the flex gadget changes monotonically 
it suffices to consider the extremes $x = 2$, $y = \frac{1}{2}$ and $x = \frac{1}{2}, y = 2$.

If $x = 2$, then $|UP| = \sqrt{5} \approx 2.24$ and $|PL| = \sqrt{3.75} \approx 1.94$. Then $|H_1 U_3| = 15 + \sqrt{5} - 1 \approx 16.24$ which is in the range of non-crossing configurations of $F_1$, and $|H_2 L_3| = 15 + \sqrt{3.75} -2 \approx 14.94$,
which is in the range of non-crossing configurations of $F_2$.

If $x = \frac{1}{2}$, then $|UP| = \sqrt{1.25} \approx 1.12$ and $|PL| = \sqrt{15} \approx 3.87$. Then $|H_1 U_3| = 15 + \sqrt{1.25} - 1 \approx  15.12$ which is in the
range of non-crossing configurations of $F_1$, and $|H_2 L_3| = 15 + \sqrt{15} -2 \approx 16.87$, which is in the range of non-crossing configurations of~$F_2$.
\end{proof}

\subsection{Properties of the Constructed Linkage}
\label{sec:linkages-wrapup}
From an \ETRINV formula $\varphi$ the construction described above produces an angle-constrained linkage $L$ together with a combinatorial embedding $\Lambda$ of $L$ (as shown in the figures above).

\begin{theorem}\label{thm:ETRINVlinkage}
There is a polynomial-time algorithm that takes as input an \ETRINV\ formula $\varphi$ and constructs 
an angle-constrained linkage 
$L$ and a combinatorial embedding $\Lambda$ of $L$ such that 
\begin{itemize}
\item if $\varphi$ is true, then $L$ has a non-crossing realization with combinatorial embedding $\Lambda$ and
\item if $\varphi$ is false, then $L$ has no realization (whether non-crossing or crossing). 
\end{itemize}
\end{theorem}
\begin{proof}
If $\varphi$ is true, then, as justified by Claims~\ref{claim:A-frame}--\ref{claim:inversion},  the values of the variables that satisfy $\varphi$ determine channel widths that result in a non-crossing realization of $L$ with combinatorial embedding~$\Lambda$.

Now suppose $\varphi$ is false. Although Hart's A-frame and the flex gadget (in isolation) have crossing realizations, our channel-restricting gadgets from Section~\ref{sec:channelrestrictgadget} ensure that the initial flex gadgets only take on non-crossing configurations that correspond to a value of a variable in the interval $[\frac{1}{2}, 2]$.  After that, the rest of our construction faithfully implements the addition and inversion constraints of~$\varphi$. 
(More formally, this is by induction.)
Therefore there is no realization of $L$, even allowing crossing configurations.    
 
\end{proof}

In the next section we give a realization-preserving transformation from angle-constrained 
linkages produced by the above construction to penny graphs.


\section{From Linkages to Penny Graphs}
\label{sec:penny-graphs}

In this section we show how to turn the angle-constrained linkages
constructed in Section~\ref{sec:linkages} into penny graphs. At a first glance that would appear to be straightforward: simply replace each vertex with a penny and each bar of integer length $\ell$ between two vertices with $\ell+1$ pennies lying on a line with the centers of the first and last pennies representing the end vertices of the bar. The centers will then have distance $\ell$ times the diameter of the pennies. An immediate obstacle is that a penny can touch at most six other pennies, so our linkages better have max-degree at most $6$, which fortunately they do. We also need to 
keep the pennies representing a bar on a line. This can be achieved by bracing the line of pennies on either (or both) sides as necessary, see the left illustration in Figure~\ref{fig:bracedbar}; the illustration on the right shows how the same bar can connect to further bars if the angles are part of the triangular grid, i.e., a multiple of $60^\circ$.
\begin{figure}[htb]
    \centering
        \begin{tabular}{cp{0.1in}c}
    \includegraphics[width=.4\textwidth]{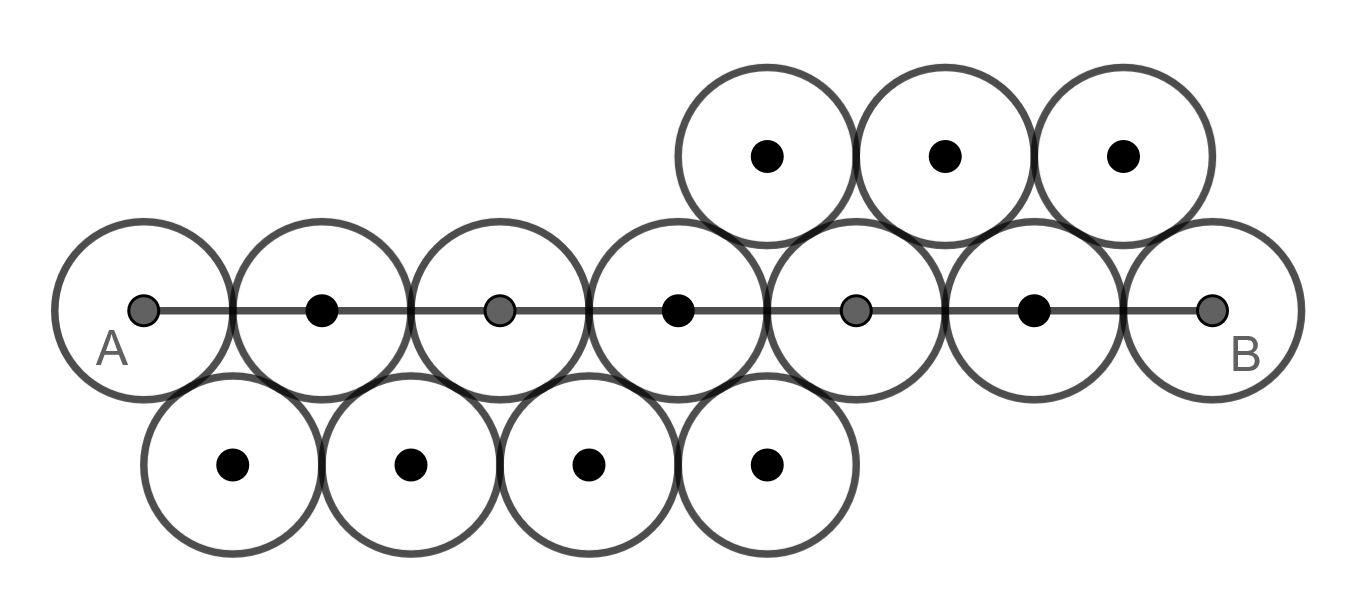}
    &&
        \includegraphics[width=.49\textwidth]{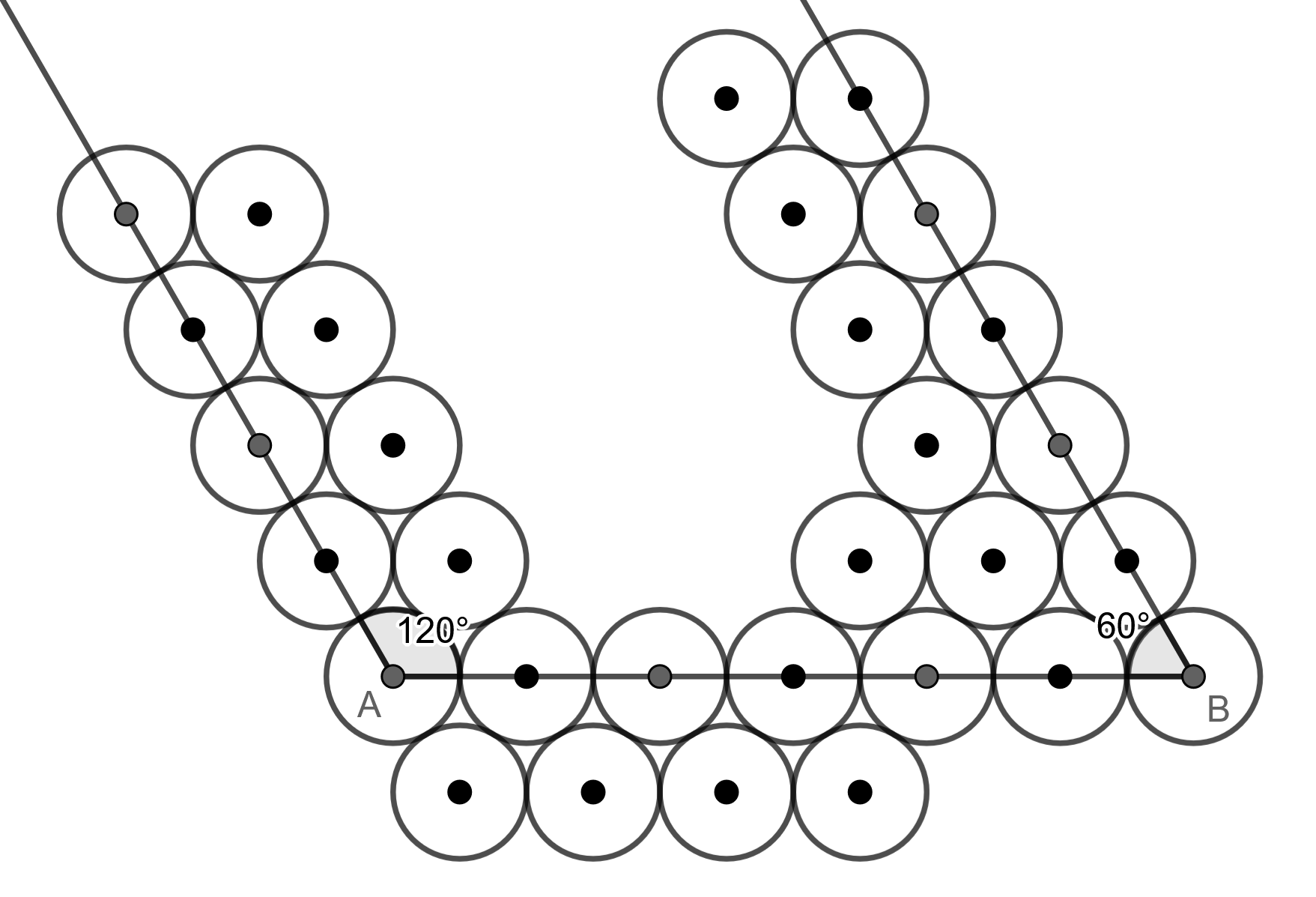}
    \end{tabular}
    \caption{{\em (left)} A bar of length $6$ (times the penny diameter) between vertices $A$ and $B$ braced on opposite sides. {\em (right)} The same bar with additional bars attached at $A$ and $B$ at angles of $120^{\circ}$ and $60^{\circ}$.  
    }
    \label{fig:bracedbar}
\end{figure}

Using extra pennies to brace a bar creates a follow-up problem: the penny at the end of a bar must be in contact with one of the extra bracing pennies.
This further restricts the degree of linkage vertices, though it is possible for two bars to share a bracing penny, see, for example, the 
$60^{\circ}$ angle at vertex $B$ in the right illustration in Figure~\ref{fig:bracedbar}.
Luckily, the linkages we constructed in Section~\ref{sec:linkages} have small degree and most of the vertex types are easy to model. The only difficult case is the flex gadget vertex that is the apex of Hart's A-frame (see vertex $H$ in 
Figure~\ref{fig:flex}). This vertex has degree $3$ and must allow a range of motion for the three angles. 
Even worse, in the modified flex gadget that has the channel-restricting gadget attached to it, the vertex $H$ (see Figure~\ref{fig:exact-interval}) 
has degree $4$ and includes one angle constrained to $180^\circ$ and three angles that must allow some range of motion.   
The gadget achieving this flexibility at the apex vertex is 
at the core of our reduction.

We need to pay close attention to the vertex types (degree and angle constraints) in our linkage gadgets. The construction in Section~\ref{sec:linkages} uses the following five gadgets: 
the {\em flex gadget} of Section~\ref{sec:flexgadget}, which incorporates Hart's A-frame from Section~\ref{sec:Aframe};
the {\em channel-restricting gadget} of Section~\ref{sec:channelrestrictgadget};
the {\em turn gadget} of Section~\ref{sec:turngadget};
the {\em cross-over gadgets} of Section~\ref{sec:crossgadgets};
the {\em addition gadget} of  Section~\ref{sec:addgadget}; and
the  {\em inversion gadget} of Section~\ref{sec:inversiongadget}.

We distinguish between \defn{fixed-angle} vertices that have all their incident angles constrained and \defn{flexible-angle} vertices 
that do not have all their incident angles constrained. We classify fixed-angle vertices in Table~\ref{tab:fixedangle} and flexible-angle vertices in Table~\ref{tab:flexangle}.

To shorten the classification, we will only include examples from one
or two gadgets for each vertex type. For example, the flex gadget contains degree-$2$ vertices at a right angle, $T$, $R$ and $S$ in Figure~\ref{fig:flex}; the same type of vertex occurs in the cross-over gadget, e.g., as $X$ in Figure~\ref{fig:crossing-same-gadget}, and 
also in the channel-restricting gadget, the addition gadget and the inversion gadget, but we do not list those instances separately in Table~\ref{tab:fixedangle}. 

\begin{table}[htb]
\setlength{\fboxsep}{3pt}
\setlength{\fboxrule}{0pt}
\centering
\begin{tabular}{|c|l|l|l|}
\hline
degree & angle(s) & pic &  sample instance(s) \\ \hline
\multirow{3}{*}{2} & $90^{\circ}$ 
& \fbox{\includegraphics[height=0.4in]{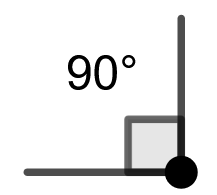}}
& $T$, $R$ in the flex gadget,  Fig.~\ref{fig:flex} \\ \cline{2-4}
                        &$120^{\circ}$ 
                        &\fbox{\includegraphics[height=0.4in]{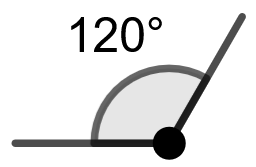}}
                        & $L$ in the addition gadget,  
                        Fig.~\ref{fig:addition-gadget} \\ \cline{2-4}
                        & $150^{\circ}$
                        &\fbox{\includegraphics[height=0.4in]{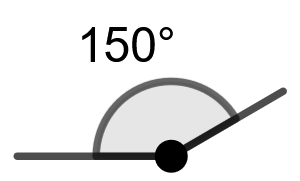}}
                        & $V$ in the addition gadget,  
                         Fig.~\ref{fig:addition-gadget}\\ \hline
\multirow{3}{*}{3} & $90^{\circ}, 90^{\circ}$ & \fbox{\includegraphics[height=0.4in]{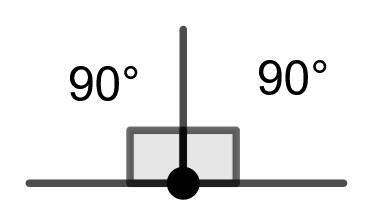}}
& $Q$ in the flex gadget,  Fig.~\ref{fig:flex}\\ \cline{2-4}
& $90^{\circ}, 120^{\circ}$& \fbox{\includegraphics[height=0.4in]{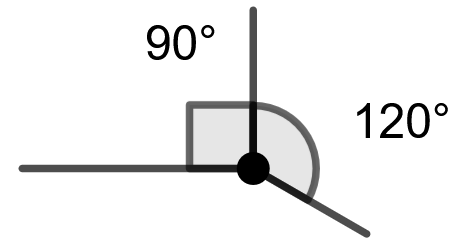}}
& $U$ in the cross-over gadget,  Fig.~\ref{fig:crossing-same-gadget}\\ \cline{2-4}
& $150^{\circ}$, $30^{\circ}$& \fbox{\includegraphics[height=0.4in]{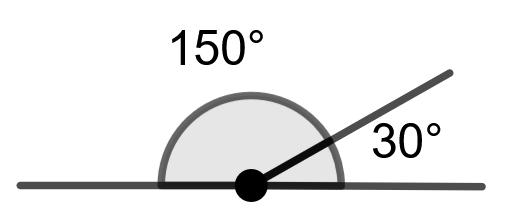}}
& $W$ in the turn gadget, Fig.~\ref{fig:crossing-gadget}  
\\ \hline
\multirow{1}{*}{4} & $60^{\circ}, 90^{\circ}, 30^{\circ}$ 
& \fbox{\includegraphics[height=0.4in]{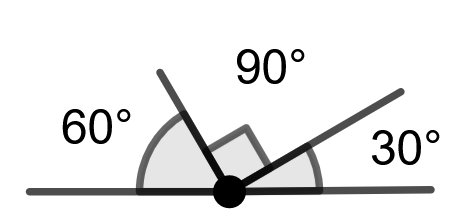}}
& $W$ in the cross-over gadget,  Fig.~\ref{fig:crossing-same-gadget} \\ \hline
    
\end{tabular}
\caption{Fixed-angle vertex types.}\label{tab:fixedangle}
\end{table}

For the flexible-angle types in Table~\ref{tab:flexangle} we must also pay attention to the range of motion needed at angles that are not constrained.  
For example, the degree-$2$ vertex $M_2$ of the channel-restricting gadget in Figure~\ref{fig:exact-interval}, must allow angles in the range $(144^\circ, 180^\circ]$. 
To simplify the classification, we overestimate the angle ranges; for example $M_2$'s range is overestimated as $(120^\circ, 240^\circ)$. 
For degree $2$ and $3$ vertices there is one degree of freedom and we only list the range of one of the angles. 
For the degree $4$ vertex $H$, the apex vertex in the modified flex gadget (see Figure~\ref{fig:A-frame-angles}), 
there is still only one degree of freedom and we list 
only the ranges for angles $\alpha$ and $\beta$ in the figure---this is explained further in the proof of Lemma~\ref{lem:angletypes}. 
Note that the case of the degree $3$ apex vertex $H$ of Figure~\ref{fig:exact-interval}
is subsumed by the degree $4$ case.

\begin{table}[htb]
\centering
\begin{tabular}
{|c|l|p{3.5in}|}
\hline
degree & angle(s) & sample instance(s) \\ \hline
\multirow{4}{*}{2} & $(60^{\circ}, 180^{\circ})$ & $C$ and $G$ in the flex gadget, 
Figs.~\ref{fig:flex},~\ref{fig:A-frame-angles};
also  $L$ in the inversion gadget, Figs.~\ref{fig:inverting-detail},~\ref{fig:inverting-detail-closeup} \\
\cline{2-3}
 & $(120^{\circ}, 240^{\circ})$ & $M_2$ and $N_2$ in the modified flex gadget,  Fig.~\ref{fig:exact-interval}; also $U$, $P$ in the inversion gadget, 
Figs.~\ref{fig:inverting-detail},~\ref{fig:inverting-detail-closeup} \\ 
\hline
\multirow{1}{*}{3}& $180^{\circ}$, $(20^{\circ},100^{\circ})$ & 
$B$, $E$, $A$ and $D$ in the flex gadget, 
Figs.~\ref{fig:flex},~\ref{fig:A-frame-angles}
\\ \hline
\multirow{1}{*}{4} & $180^{\circ}$, 
$(85^\circ,114^\circ)$, $(35^{\circ},46^{\circ})$
& $H$, the apex vertex in the modified flex gadget, Figs.~\ref{fig:exact-interval}, \ref{fig:A-frame-angles}, with ranges for angles $\alpha, \beta$ \\ \hline
\end{tabular}
\caption{Flexible-angle vertex types. Notation like $(60^\circ, 180^\circ)$ indicates the range for a single angle.
}\label{tab:flexangle}
\end{table}

\begin{lemma}
\label{lem:angletypes} Tables~\ref{tab:fixedangle} and~\ref{tab:flexangle} cover all angle types required by our linkage construction in Section~\ref{sec:linkages}. 
\end{lemma}
\begin{proof}
A close inspection of the linkage gadgets shows that all their vertex types are captured in the tables.

For Table~\ref{tab:flexangle}, there are two additional aspects to verifying the angle ranges.  First, we claim that all non-constrained angles change 
monotonically so an angle always lies between its values in the extreme positions for $x=2$ and $x = 0.5$. 
This is clear for 
degree $2$ flexible-angle vertices.  The other flexible-angle vertices only appear in the [modified] flex gadget.  Refer to Figure~\ref{fig:A-frame-angles}. 
Monotonicity 
for the degree-$3$ flexible angle vertices $B,E,A,D$ is clear.
The situation at vertex $H$ is more complicated.  As $h$ decreases, angle $\alpha$ is monotonically increasing and $\gamma$ is monotonically decreasing.   
In Appendix~\ref{app:AFrame} we show that angle $\beta$ is the same as angle $AHD$, and therefore it is monotonically increasing. 

Finally, 
for the values of the angles in the extreme positions 
we rely on the figures, but it is 
possible to justify these values algebraically. 
\end{proof}

\begin{figure}[htb]
    \centering
    \includegraphics[width=0.38\linewidth]{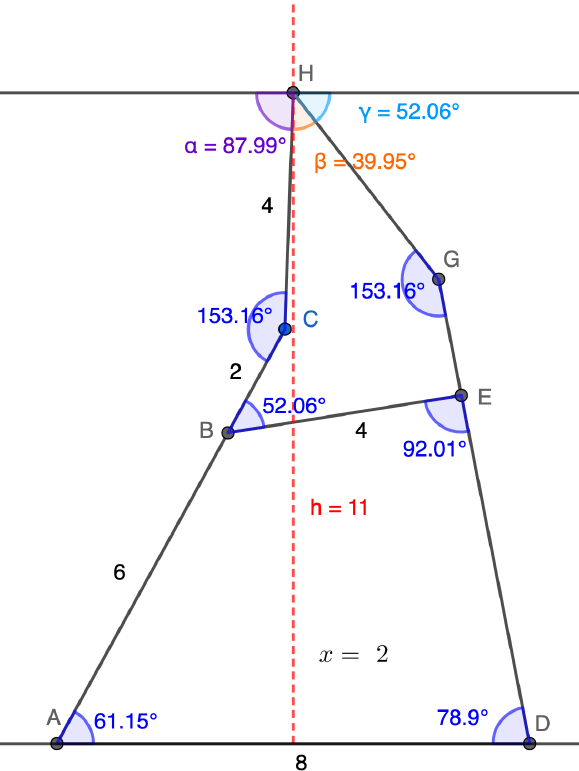}
    \hspace{.3in}
    \includegraphics[width=0.38\linewidth]{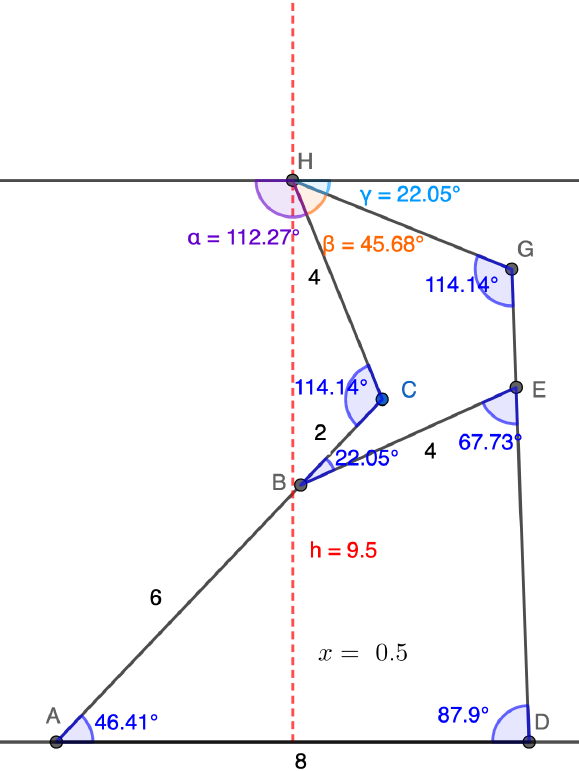}
\caption{
Angles in the flex gadget. 
\emph{(left)} The configuration realizing $x=2$.
\emph{(right)} The ``kneeling'' configuration realizing $x = 0.5$. In between these two positions, the flex gadget continuously simulates all values $x \in [\frac{1}{2},2]$ as we proved earlier.
}
    \label{fig:A-frame-angles}
\end{figure}

\begin{figure}
    \centering
\includegraphics[width=.3\textwidth]{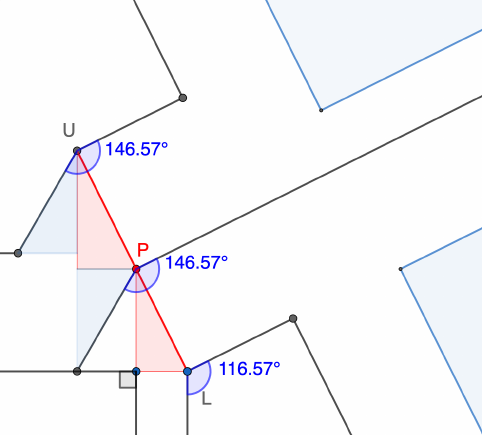}
\hskip0.4in
\includegraphics[width=.3\textwidth]{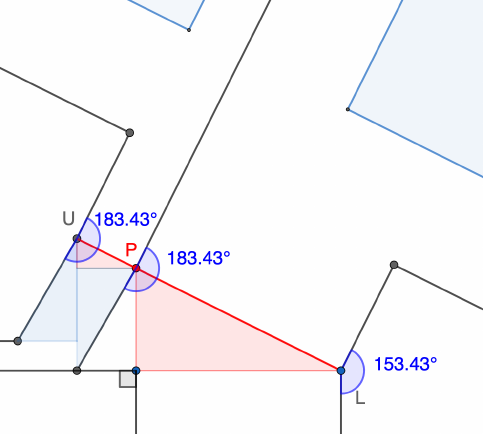}
\caption{Closeup showing the angle ranges at vertices $U, P, L$ for $x=2$ (left) and $x=0.5$ (right) from Figure~\ref{fig:inverting-detail}.}
    \label{fig:inverting-detail-closeup}
\end{figure}

The next lemma shows that we can simulate an angle-constrained linkage with angle types as identified in Tables~\ref{tab:fixedangle} and~\ref{tab:flexangle} by a penny graph. 

\begin{lemma}\label{lem:linkpenny}
There is a polynomial-time algorithm that takes as input an angle-constrained
linkage~$L$ 
as constructed in Section~\ref{sec:linkages}
(with edge-weights in $\NN/2$ and angle types as described in Tables~\ref{tab:fixedangle} and \ref{tab:flexangle}) and a combinatorial embedding $\Lambda$
and constructs a graph $G$ and a combinatorial plane embedding $D$ of $G$ such that
\begin{itemize}
    \item if $L$ is realizable, then $G$ has a penny-graph embedding realizing $D$, and
    \item if $L$ is not realizable, then $G$ is not a weak unit-distance graph for any embedding.
\end{itemize}
\end{lemma}

Bars and vertices in a linkage do not have area, but pennies do; when replacing bars with bar gadgets and vertices with vertex gadgets these gadgets take up space. Consider realizing a bar of length $\ell$. We could do so using a line of $\ell+1$ pennies of unit diameter braced on either (or both) sides. That bar has a thickness of $(\sqrt{3}+1)/2<2$ so it is less than two pennies thick (measuring thickness as distance from the middle of the bar). We can reduce the thickness of the bar to below any value $\varepsilon > 0$ by working with a line of $t\ell+1$ pennies each of diameter $1/t$. If we choose $t > 2/\varepsilon$ the resulting line of pennies (braced on either or both sides) still has length $\ell$ and thickness less than $2/t<\varepsilon$. Doing the same thing for vertex gadgets allows us to make bars sufficiently thin, and vertex gadgets sufficiently small so that they do not interfere with each other. 

More formally, we introduce the notion of ``thickness''---for the proof of Lemma~\ref{lem:linkpenny} only---to measure the space requirements of our gadgets. The \defn{thickness} of our vertex and bar gadgets is defined as the smallest $\varepsilon$ such that the pennies in the gadget lie within an $\varepsilon$-neighborhood of the linkage part they represent (vertex or bar). We do not need the exact $\varepsilon$, upper bounds are sufficient. We distinguish the thickness of a vertex (the radius of a disk containing all pennies representing the vertex) and the thickness of a bar (half the width of the bar-gadget containing all pennies representing the bar). When analyzing the gadgets, we express their thickness in terms of pennies, e.g., the bar-gadget we described earlier is less than two pennies thick; the actual thickness will then depend on the diameter we choose for the pennies. 

\begin{proof}
We start by showing how to replace vertices of each type with penny graph gadgets. 

\begin{figure}[htb]
    \centering
    \includegraphics[height=3in]{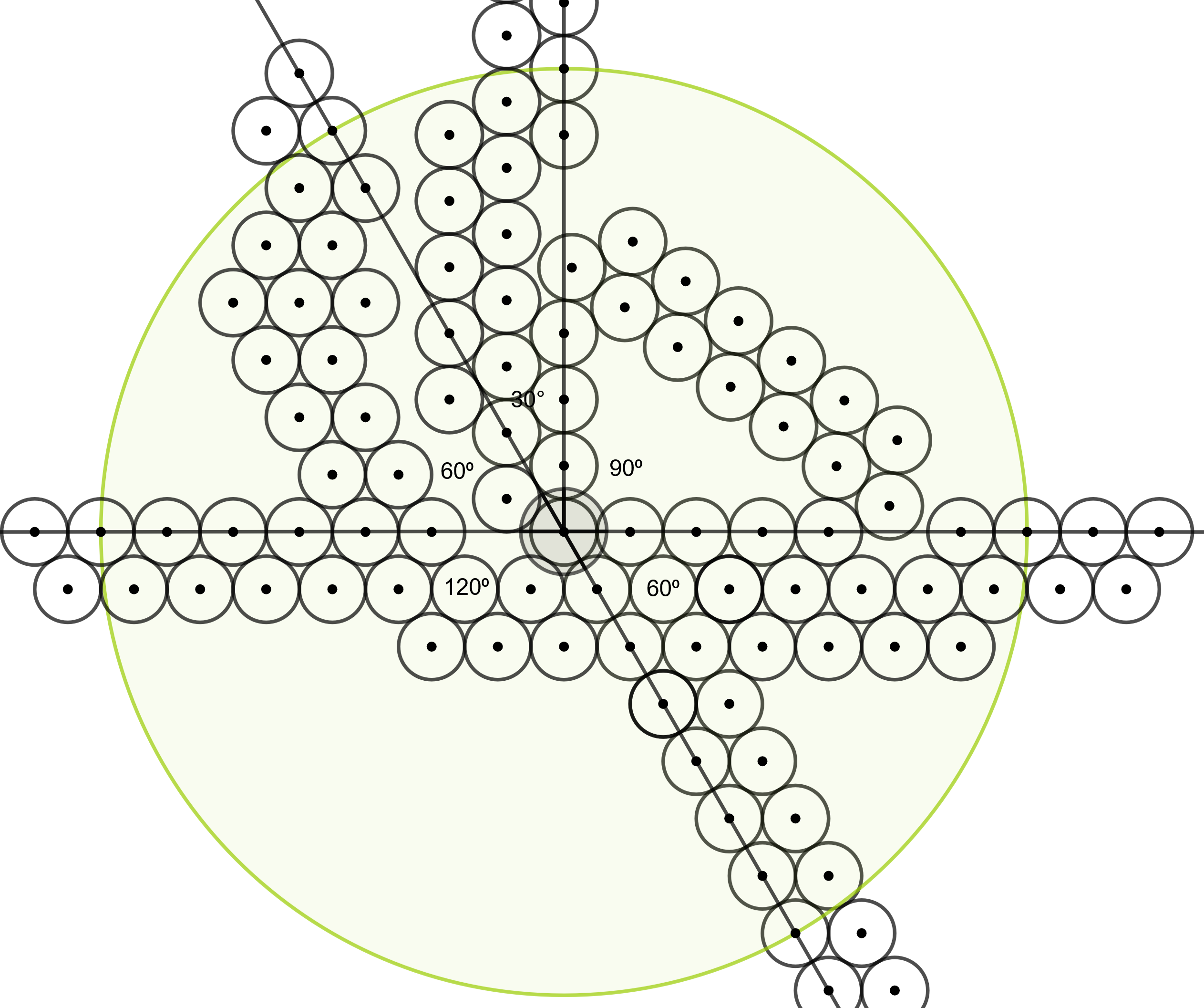}
    \caption{
    The universal vertex gadget with angles $90^{\circ}$, $60^{\circ}$, $120^{\circ}$, $60^{\circ}$, and $30^{\circ}$. The green disk shows that our gadget is at most $7$ pennies thick. 
    }
    \label{fig:vx2fix}
\end{figure}

\paragraph{Fixed-Angle Vertex Types.} 
Gadgets for $120^{\circ}$- and $60^{\circ}$-angles are easy to build as we already saw in Figure~\ref{fig:bracedbar}, since they are part of the triangular grid. A right angle can be forced using the Pythagorean triple $(3,4,5)$---this idea is based on~\cite[Figure~5]{ADDELS25}; combining these two ideas we can build a universal vertex gadget, shown in Figure~\ref{fig:vx2fix}, which realizes bars at angles $90^{\circ}$, $60^{\circ}$, $120^{\circ}$, $60^{\circ}$, and $30^{\circ}$ in this (cyclic) order. 
In this and the following figures, the thickness of a vertex gadget is shown as a green disk.  
Only pennies lying fully inside the green disk
belong to the 
vertex gadget; the remaining pennies belong to bar gadgets connecting the vertex gadget to another vertex gadget and are included for illustration only.

This universal vertex gadget can be used  to simulate all the fixed-angle vertex types from Table~\ref{tab:fixedangle} by extending the appropriate subset of arms, and (if necessary) flipping the gadget: the gadget contains bars at angles $90^{\circ}$, $120^{\circ}$ and $150^{\circ}=90^{\circ}+60^{\circ}$, 
simulating all the degree-$2$ fixed-angle vertices; it contains bars at angles $90^{\circ}=60^{\circ}+30^{\circ}$, $90^{\circ}$; $90^{\circ}=60^{\circ}+30^{\circ}$, $120^{\circ}$; and  $150^{\circ}=90^{\circ}+60^{\circ}$, $30^{\circ}$, simulating all the degree-$3$ fixed-angle vertices; and finally it contains bars at angles $30^{\circ}$, $90^{\circ}$, $60^{\circ}$, which simulate the degree-$4$ fixed-angle vertex. 
The universal vertex gadget is at most $7$ pennies thick.

\begin{figure}[htb!]
    \centering
    \includegraphics[width=.6\textwidth]{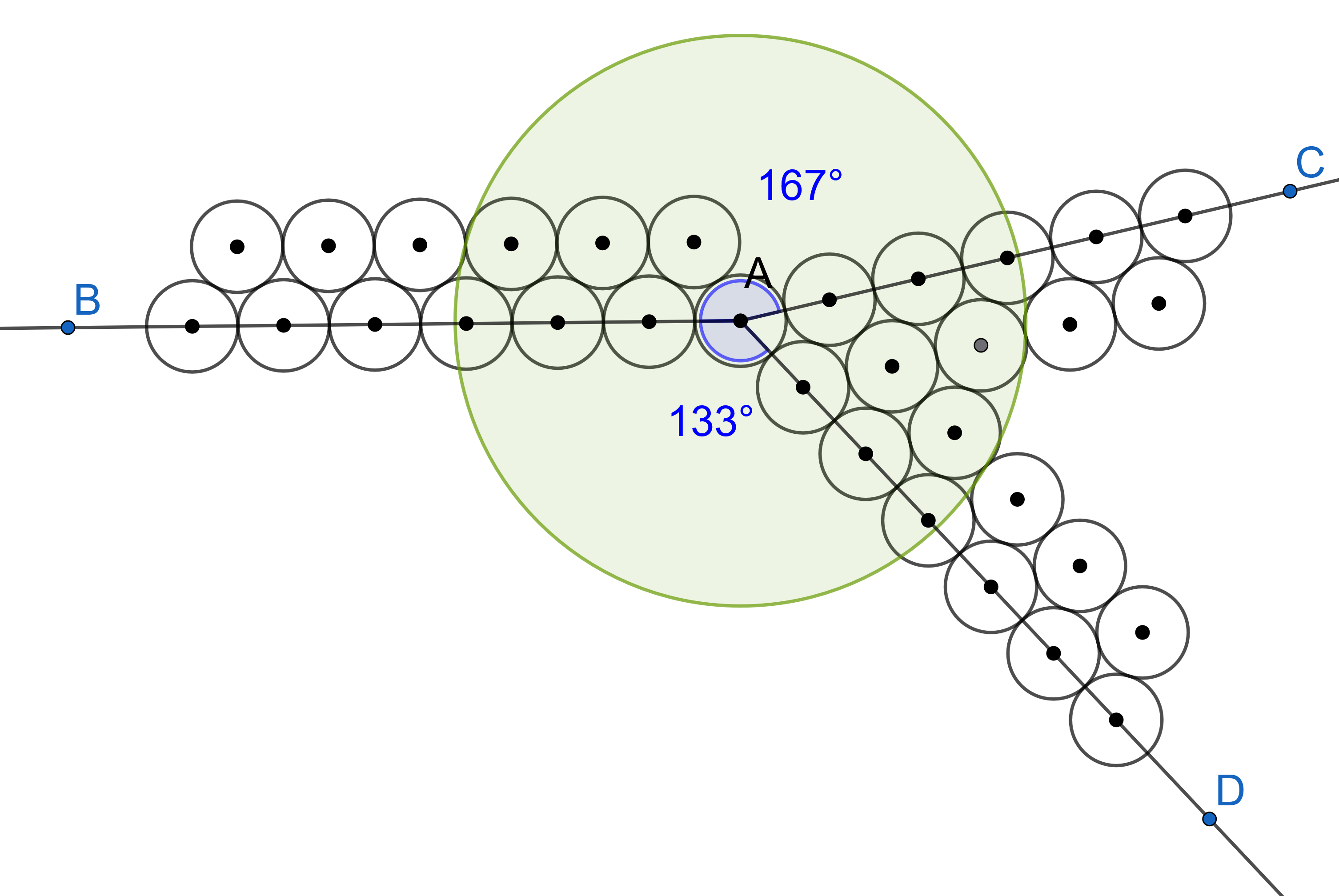}
    \caption{
    A universal flexible degree-$2$ vertex gadget; the angle between arms $BA$ and $AC$ can range in  the open interval $(120^{\circ}, 240^{\circ})$, and the angle between arms $BA$ and $AD$ in the open interval $(60^{\circ}, 180^{\circ})$. The green disk indicates that the gadget is at most $3$ pennies thick.
    \Gapp degree-2-flexible-vertex.ggb; point $C$ can be moved.
    }
    \label{fig:vx2flex180}
\end{figure}

\paragraph{Flexible-Angle Vertex Types.} 
For these vertex types we need more than one gadget. For the flexible-angle degree-$2$ vertices, we use the universal gadget in Figure~\ref{fig:vx2flex180}. Arms $BA,AD$ simulate the bars with angles in the range $(60^{\circ}, 180^{\circ})$ while arms $BA$,$BC$ realize the range $(120^{\circ}, 240^{\circ})$.

The 
flexible degree-$3$ vertex has a fixed angle of $180^{\circ}$ and a flexible angle; the vertex gadget in Figure~\ref{fig:vx3flex180} shows that we can simulate a flexible angle in the range $[20^{\circ}, 100^{\circ}]$, indeed the actual range is even larger. 

\begin{figure}[htb]
   \begin{tabular}{cp{0.01in}c}
    \includegraphics[width=.47\textwidth,valign=t]{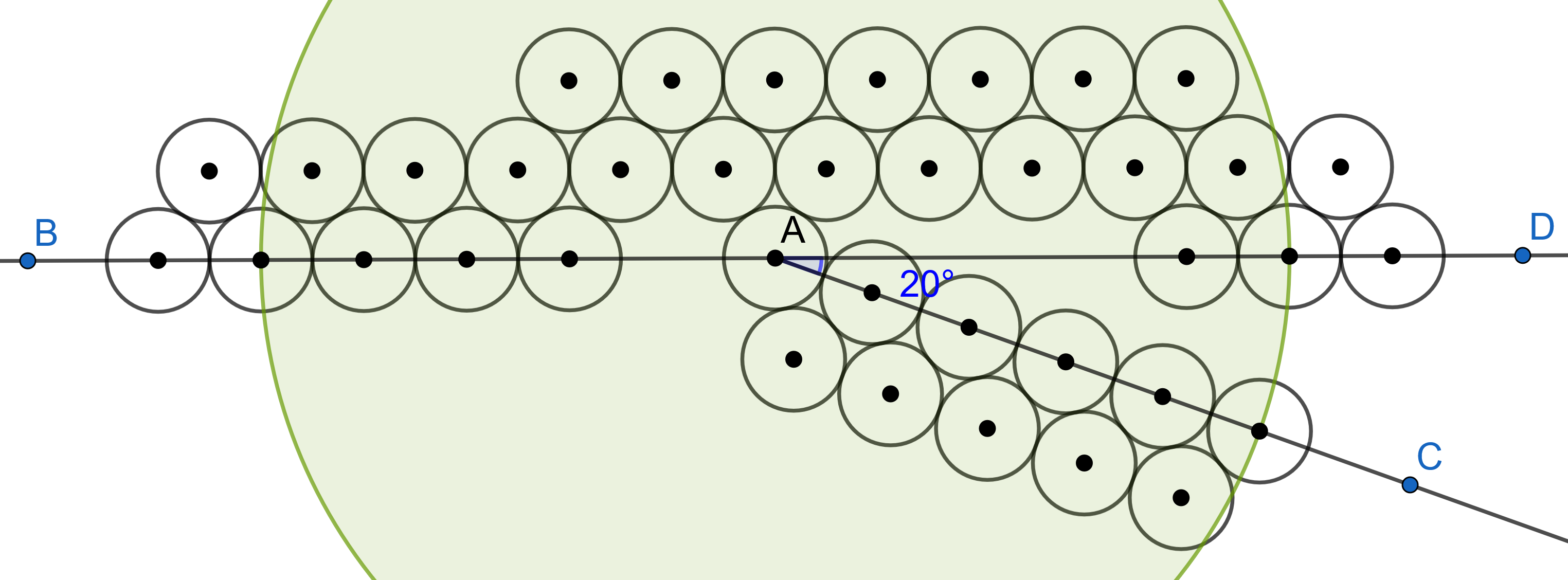}
    &
    \includegraphics[width=.47\textwidth,valign=t]{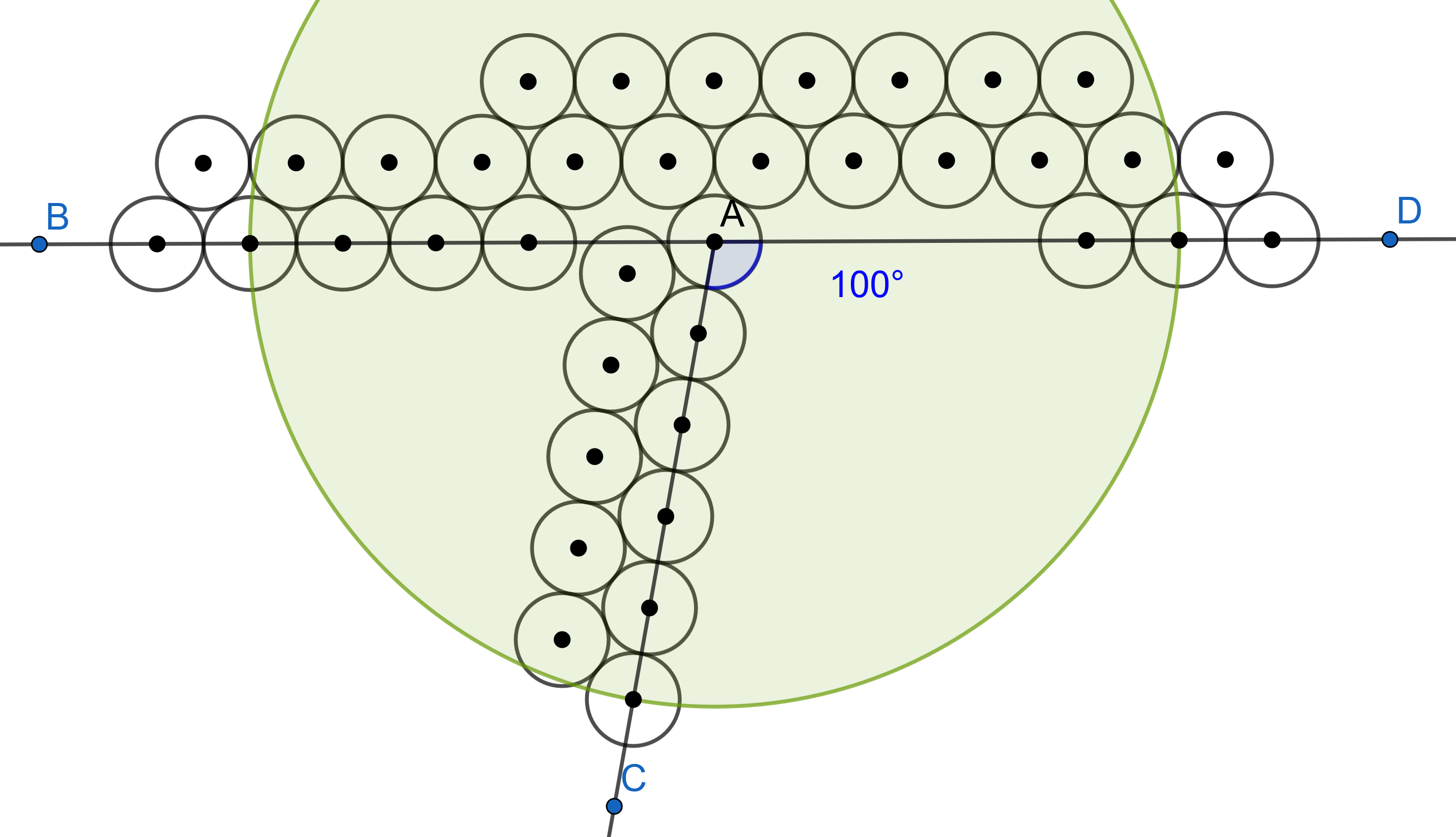}
    \end{tabular}
    \caption{
    A flexible degree-$3$ vertex gadget with a straight-line bar (through $AB$) and a flexible angle $\angle BAC$ at $20^{\circ}$ {\em (left)} and at $100^{\circ}$ {\em (right)}. The angle can range in the closed interval $[20^{\circ}, 100^{\circ}]$. The green disk shows that the gadget is at most $5$ pennies thick.
    \Gapp degree-3-flexible-vertex.ggb; point $C$ can be moved.
    }
    \label{fig:vx3flex180}
\end{figure}

This leaves us with the trickiest case, the apex vertex $H$ of the modified flex gadget; 
 see
Figure~\ref{fig:A-frame-angles}, which shows $H$ and its angles $\alpha, \beta$, and $\gamma$ in 
the positions for $x=2$ and $x=0.5$.
From Table~\ref{tab:flexangle} and the proof of Lemma~\ref{lem:angletypes} we have
that $\alpha$ and $\beta$ are monotonically increasing in the ranges $(85^\circ, 114^\circ)$ and $(35^\circ, 46^\circ)$, respectively.
There is a single degree of freedom.

At the heart of the apex vertex gadget is a $4\times 4$ grid of pennies---colored blue in the figures, see Figure~\ref{fig:panto}. We call one corner vertex of this grid $H$ and brace the two sides not incident to $H$ with three pennies each, so that the blue pennies are forced to realize a rhombus; the angle of the rhombus at $H$, without the additional appendages, can range in the interval $(60^{\circ},120^{\circ})$.\footnote{The grid of pennies has similarities with the {\em pantograph}, one of the oldest known linkages, dating to the early 17th century, see~\url{https://historyofinformation.com/detail.php?id=1259}.}
We extend the two sides incident to $H$ into arms as shown in the figure; one arm connects to $G$, the other arm is a virtual arm, colored gray in the illustrations; the angle between the virtual arm and $HG$ cannot be made sufficiently small for our purposes, so we extend the virtual arm, outside the area of the rhombus, so it can simulate an arm rotated $60^{\circ}$ counterclockwise at $H$. We use this new arm to connect to $C$. Figure~\ref{fig:panto} shows that this gadget can realize angles $CHG$ in the range 
$[20^{\circ},50^{\circ}]$
which accommodates the range of angle $\beta \in (35^\circ, 46^\circ)$.

Finally, we attach $H$ to a subset of a triangular grid that simulates a line through $H$; this gives us the final apex vertex gadget, 
as shown in Figures~\ref{fig:apexs} and~\ref{fig:apexk}, with the line through $H$ shown as horizontal. 
Figure~\ref{fig:apexs} shows that the apex vertex gadget can realize the 
minima $\alpha = 85^\circ$ and $\beta = 35^\circ$.
Figure~\ref{fig:apexk} shows that the apex vertex gadget can realize the 
maxima $\alpha = 114^\circ$ and $\beta = 46^\circ$.
Since $\alpha+\beta$ is monotonically increasing and every value of $\beta$ in its range can be realized, 
this implies that every configuration of $H$ in the modified flex gadget between the configurations for $x=0.5$ and $x=2$ can be realized by the apex vertex gadget.

\begin{figure}[htb]
   \begin{tabular}{cp{0.01in}c}
    \includegraphics[width=.47\textwidth,valign=t]{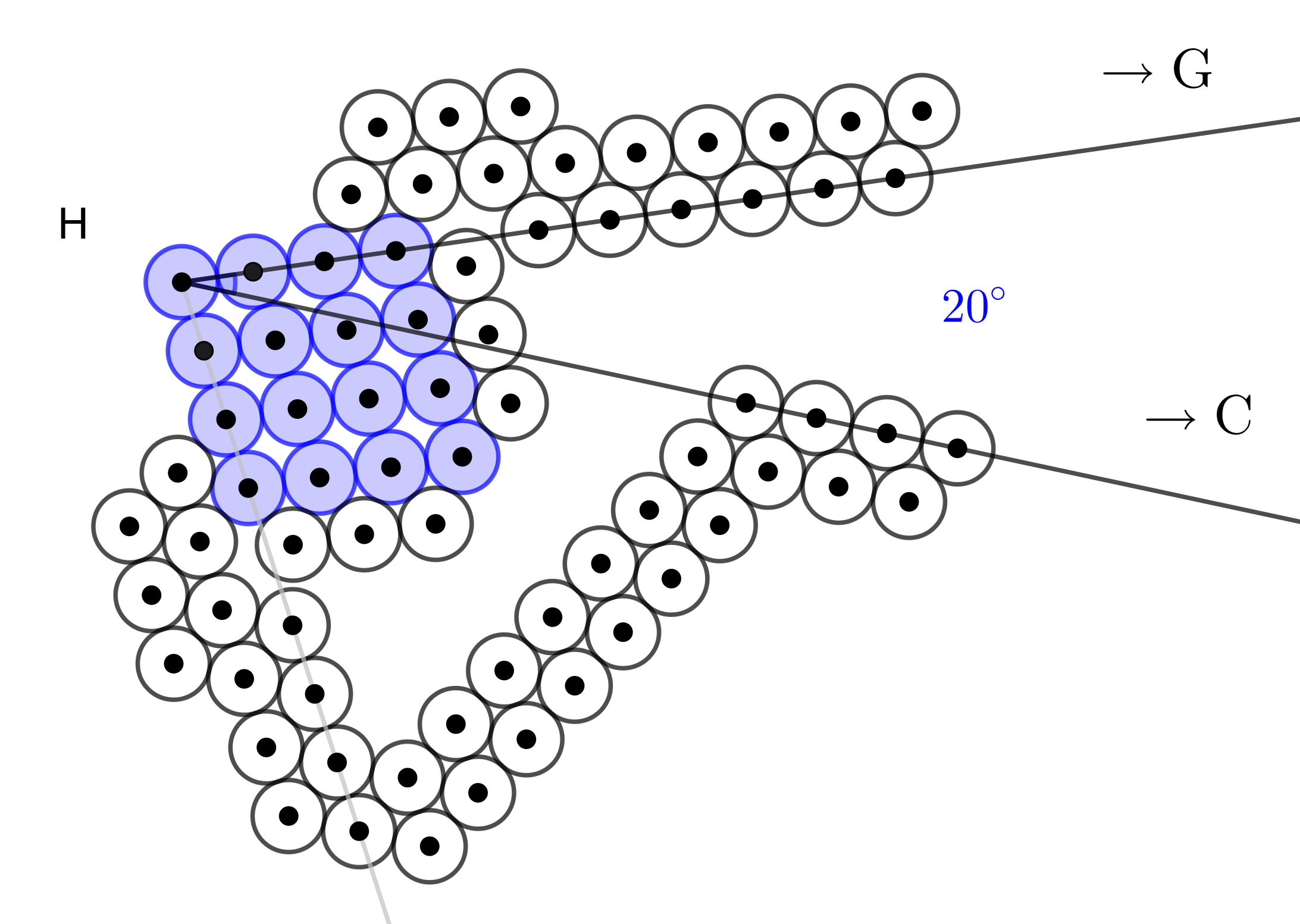}
    &
    \includegraphics[width=.47\textwidth,valign=t]{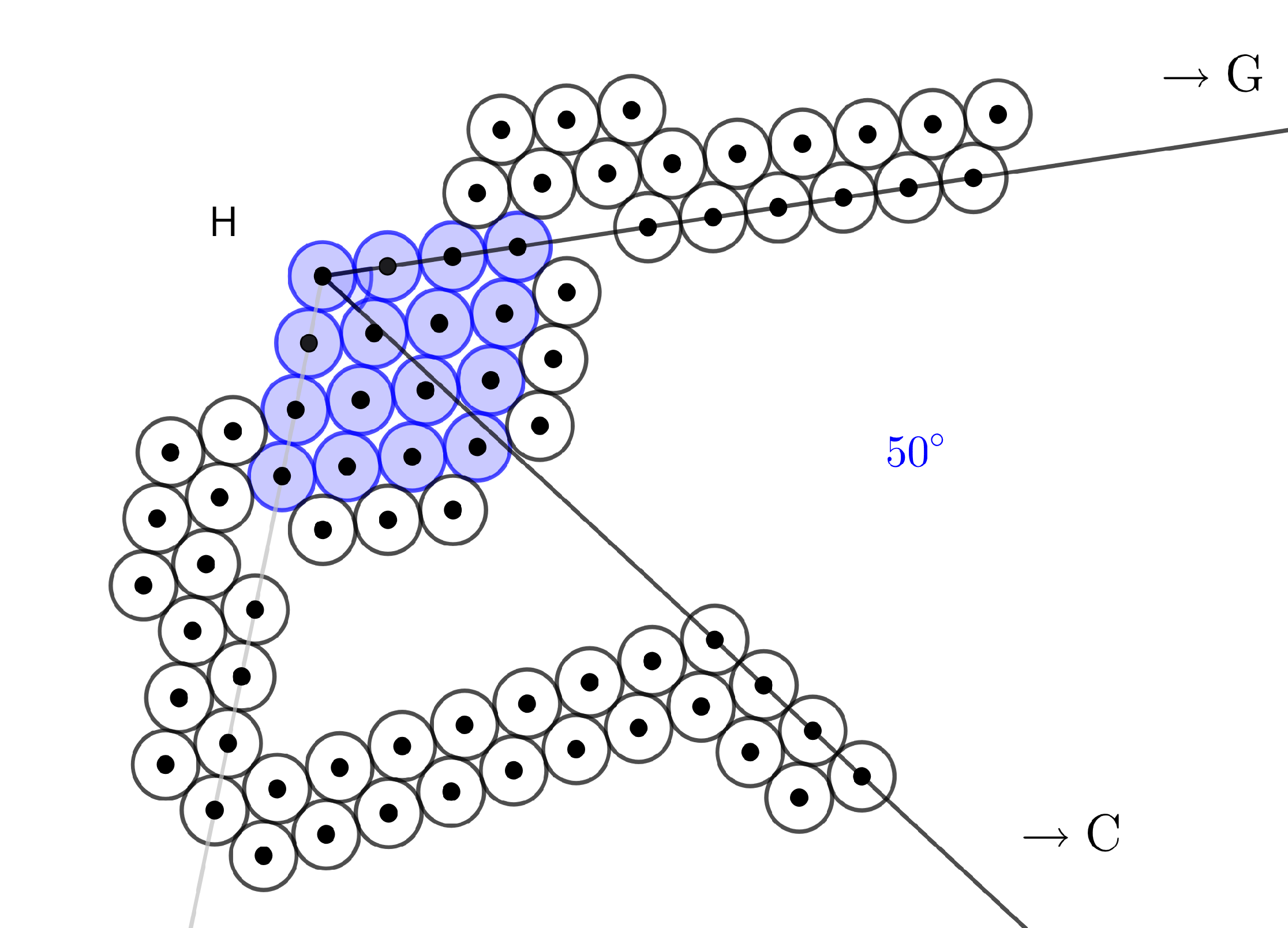}
    \end{tabular}
    \caption{
      The heart of the apex vertex gadget is a $4 \times 4$ grid of pennies (shown in blue) which 
      realizes angle $CHG$ in the range $[20^\circ, 50^\circ]$.  
    }
    \label{fig:panto}
\end{figure}

\begin{figure}[htb]
\centering
    \begin{tabular}{c}
    \includegraphics[width=.63\textwidth]{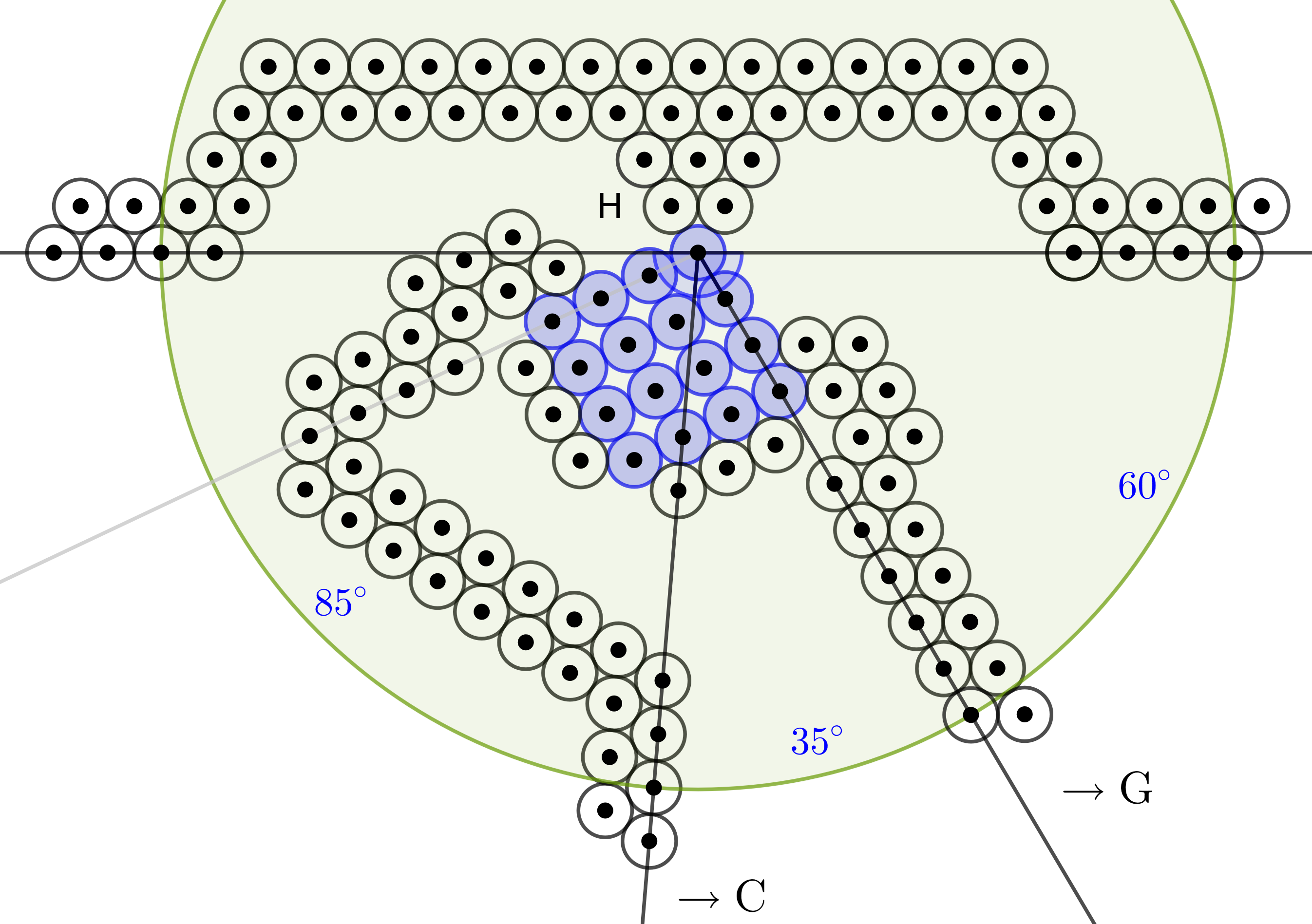}
    \end{tabular}
    \caption{
    Adding the horizontal bar to create the final apex vertex gadget, shown here close to the
 standing position corresponding to $x = 2$, see Fig.~\ref{fig:A-frame-angles}(\emph{left}).
 The green disc shows that the gadget is less than $11$ gadgets thick. 
     \Gapp degree-4-flexible-vertex.ggb; the two points  in the $4 \times 4$ grid closest to $H$ can be moved.
    }
    \label{fig:apexs}
\end{figure}

\begin{figure}[htb]
\centering
    \includegraphics[width=.7\textwidth]{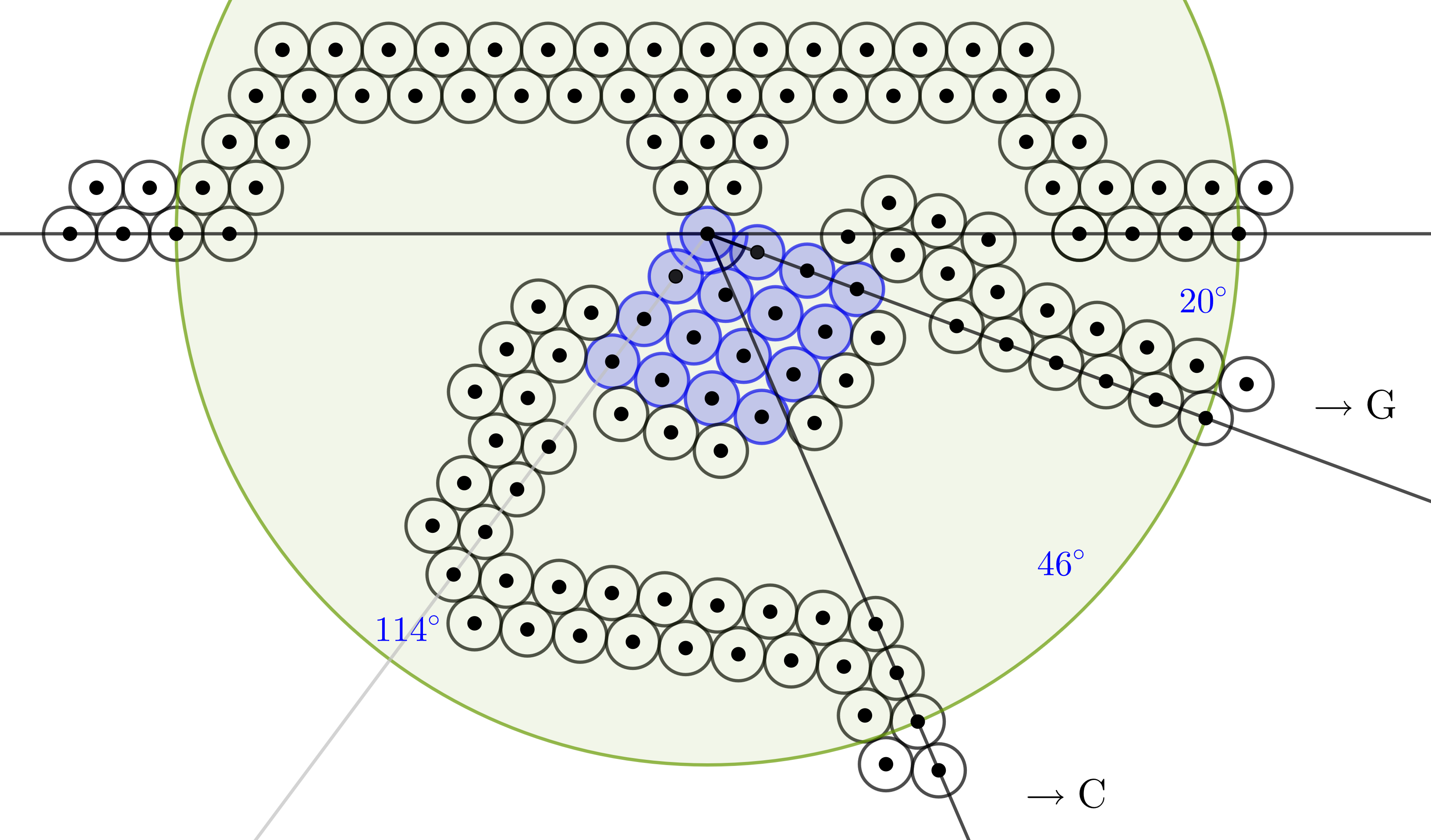}
    \caption{
    The apex vertex gadget close to  the kneeling position corresponding to $x = 0.5$, see Fig.~\ref{fig:A-frame-angles}(\emph{right}).      
    }
    \label{fig:apexk}
\end{figure}

\paragraph{Completing the Construction.}
For our \ER-hardness proof we can assume, by  Theorem~\ref{thm:ETRINVlinkage},  that we are given an angle-constrained linkage $L$ and a combinatorial embedding $\Lambda$. We know that if $L$ can be realized, then we can assume that $L$ realizes $\Lambda$ and that in this realization the minimum distance between any two vertices (the \defn{vertex-vertex resolution}) and any vertex and an edge not incident to the vertex (the \defn{vertex-edge resolution}) is 
at least $\varepsilon > 0$, where $\varepsilon$ is a fixed constant greater than $0$ independent of $L$.

At this point we have shown how to realize each vertex type in $L$, (Lemma~\ref{lem:angletypes} showed that the list is complete).
All our penny graph gadgets are less than $11$ pennies thick (which is needed to simulate the apex vertex in Hart's A-frame). Let $t = \lceil 11/\varepsilon \rceil$. We build a graph $G$ and an embedding $D'$ from $(L,D)$ as follows: replace each vertex of $L$ with the penny graph gadget we constructed. For every bar $uv \in L$ we do the following: 
suppose $uv$ has length $\ell \in \NN/2$ (all our distances are half-integers; actually, with one exception, they are all integers); connect the appropriate arms of the vertex gadgets of $u$ and $v$ by a braced arm of $t \ell+1$ pennies (merging the arms into the vertex gadgets at its ends). Let $G$ be the resulting contact graph of the pennies, and $D$ the combinatorial 
embedding of $G$ based on the combinatorial embedding $\Lambda$ of $L$ and the (geometric) embeddings of the vertex gadgets we have seen. It follows that if $L$ is realizable with combinatorial embedding $\Lambda$, then $G$ is realizable as a penny graph with combinatorial embedding $D$. If we work with pennies of diameter $1/t$, then a bar $uv$ of length $\ell$ in $L$ will also have length $t \ell/t = \ell$ in the realization of $G$. We picked $t$ sufficiently large so that no vertex gadget can interfere with another vertex gadget, or an edge gadget which it is not incident to. Finally, any realization of $G$ (whether realizing $D$ or not) can be turned into a realization of $L$, simply because the penny graph gadgets enforce the proper distances between vertices. This completes the proof. 
\end{proof}

\clearpage

\section{Marbles and Ball Contact Graphs}
\label{sec:marbles}

Coin and penny graphs generalize to higher dimensions: A graph $G = (V,E)$ is a \defn{ball contact graph} in $\RN^d$, $d\geq 2$ if 
there is a mapping from vertices $v \in V$ to center points 
$c(v) \in \RN^d$ and radii 
$r(v) \in \RN_{>0}$,
such that 
\begin{itemize}
    \item if $uv \in E$, then $d(u,v) = r(u)+r(v)$, and
    \item if $uv \not\in E$, then $d(u,v)>r(u)+r(v)$,
\end{itemize}
for all pairs of points $u,v \in V$, where $d(u,v)$ is the Euclidean distance of $u$ and $v$.
 In plain English, there is a family of closed balls $B^d(c(v),r(v))$, $v \in V$, such that the balls corresponding to the endpoints of an edge touch, and all other pairs of balls do not overlap. If all radii are the same (equivalently all radii are $1$), we speak of \defn{unit-ball contact graphs}; unit-ball contact graphs in $\RN^3$ are also known as \defn{marble graphs}. 

Since coin graphs coincide with the planar graphs, they are easy to recognize; this fails in higher dimensions: contact graphs of balls in $\RN^3$ are \NP-hard to recognize, a result attributed to Kirkpatrick and Rote by Hlin{\v e}n\'y and Kratochv\'il~\cite[Corollary 4.6]{HK01}. The result follows from a reduction by Kirkpatrick and Rote~\cite[Proposition 4.5]{HK01}: a graph $G$ is a unit-ball contact graph in $\RN^d$ if and only if $G + K_2$ is a ball contact graph in $\RN^{d+1}$ for every $d \geq 2$. Here $+$ denotes the join of two graphs. Applying this reduction to Theorem~\ref{thm:penniesER} immediately yields the following result.

\begin{corollary}\label{cor:ball3ER}
    Recognizing ball contact graphs in $\RN^3$ is \ER-complete. 
\end{corollary}

The reduction by Kirkpatrick and Rote turns unit balls in $\RN^d$  into 
(not necessarily unit)
balls in $\RN^{d+1}$; we are not aware of a general construction that lifts a unit ball result in $\RN^d$ to a unit ball result in $\RN^{d+1}$, but we
present below
a reduction from $\RN^2$ to $\RN^3$ which yields our second main result: the recognition problem of marble graphs is \ER-complete. 

\begin{theorem}\label{thm:marblesER}
 Recognizing marble graphs is \ER-complete.
\end{theorem}

As an immediate consequence of Theorem~\ref{thm:marblesER} and the Kirkpatrick-Rote reduction we obtain the following result. Hlin{\v e}n\'y and Kratochv\'il~\cite{HK01} had shown \NP-hardness.

\begin{corollary}\label{cor:ball4ER}
    Recognizing ball contact graphs in $\RN^4$ is \ER-complete. 
\end{corollary}

\begin{remark}\label{rem:BCGhigher}
Hlin\v en\'y~\cite{H97} and Hlin\v en\'y and Kratochv\'il~\cite{HK01} proved \NP-hardness of recognizing unit-ball contact graphs in dimensions $3$, $4$, $8$ and sketched an approach for $24$ (which implies \NP-hardness of recognizing ball contact graphs in dimensions $4$, $5$, $9$ and $25$ by the Kirkpatrick-Rote reduction). It is possible that their ideas can be combined with our reduction. The difficulty appears to lie in constructing rigid gadgets in higher dimensions. 
\end{remark}

We base our proof of Theorem~\ref{thm:marblesER} on Theorem~\ref{thm:penniesER}, but we should point out that the weaker result that unit-distance recognition is \ER-complete~\cite{S13} would be sufficient as a starting point.

\begin{proof}[Proof of Theorem~\ref{thm:marblesER}]
    Membership is immediate from the definition of marble graphs. For \ER-hardness, we reduce from the penny graph recognition problem. For a given graph $G$ we build a graph $H$ such that if $G$ is a penny graph then $H$ is a marble graph, and if $H$ is not a marble graph, then $G$ is not a weak unit distance graph. This is sufficient by Theorem~\ref{thm:penniesER}.

For intuition about our reduction, imagine a penny graph realization.  Each vertex in the plane will be replaced by a vertical column in the $z$-direction to which we attach arms for each incident edge; these arms are assigned to different $x$-$y$ layers so they do not interfere with each other.
The idea of stacking bars was motivated by actual physical linkages.\footnote{See \url{https://americanhistory.si.edu/collections/search?edan_q=inversor} for examples.}
To realize this plan we will need various rigid marble graphs in three dimensions.

    We will visualize the realization of a marble graph by displaying each vertex at the center of the marble representing it, e.g., $K_4$ would be visualized as a tetrahedron. 
    The core structure we work with consists of a flat hexagon on which we erect three tetrahedra whose apexes form a triangle; we call this a \defn{UFO}, see Figure~\ref{fig:smallUFO}.\footnote{The convex hull is a triangular cupola.}
    \begin{figure}[htbp]
     \centering
     \includegraphics[height=2in]{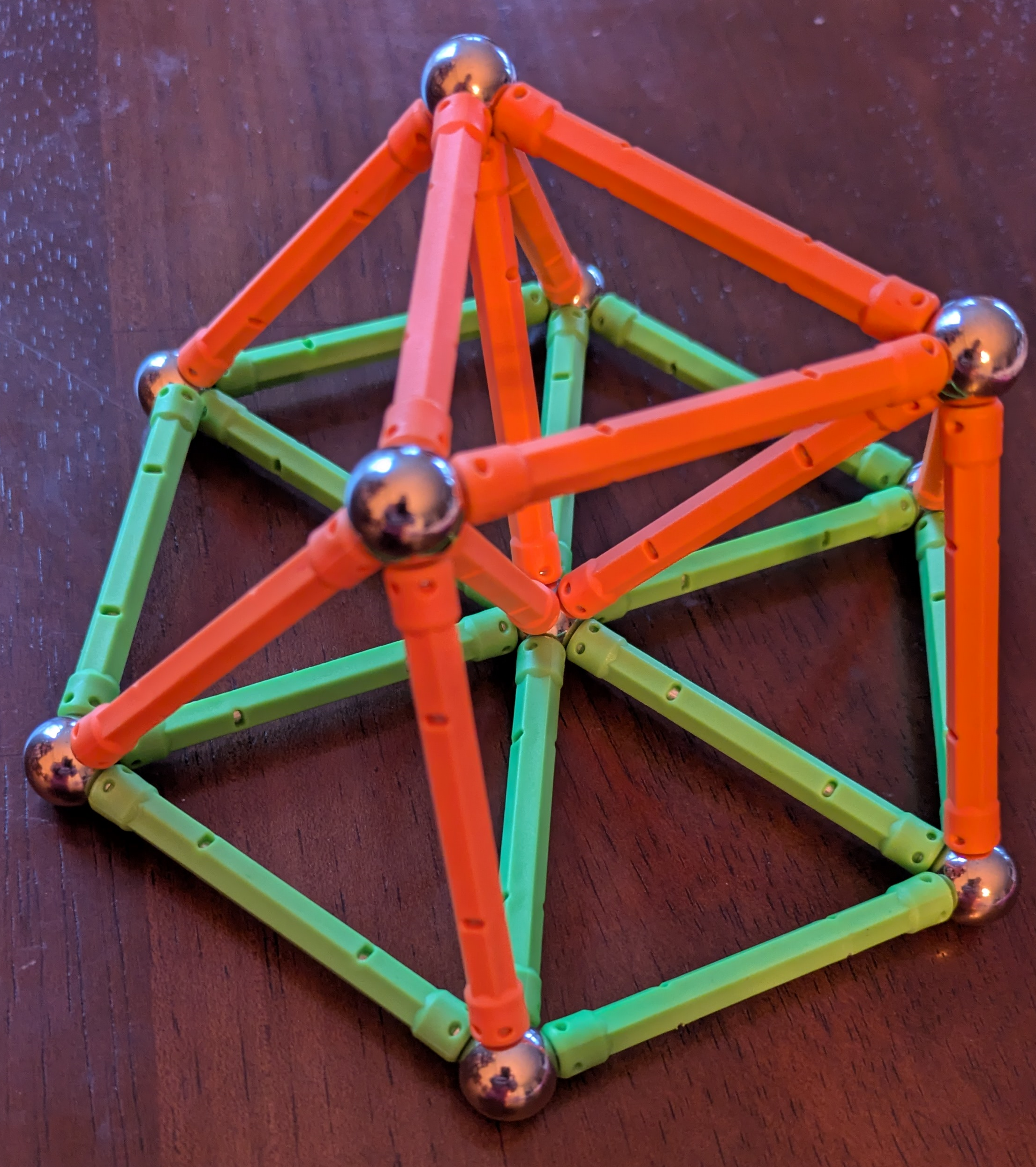}
     \caption{The UFO, a rigid marble graph.}
        \label{fig:smallUFO}
    \end{figure}
    
    More generally, for any integer~$k$ we can arrange $k+1$ UFOs along a line, such that for each consecutive pair their base hexagons share two triangles and 
    the two UFOs overlap in one tetrahedron.
    The centers of the hexagons of the first and last UFO have distance $k$, see Figure~\ref{fig:karm} (note that we have to add an edge between two overlapping UFOs in the top triangle layer). We call this a \defn{$k$-arm}.
     \begin{figure}[htbp]
     \centering
     \includegraphics[height=2in]{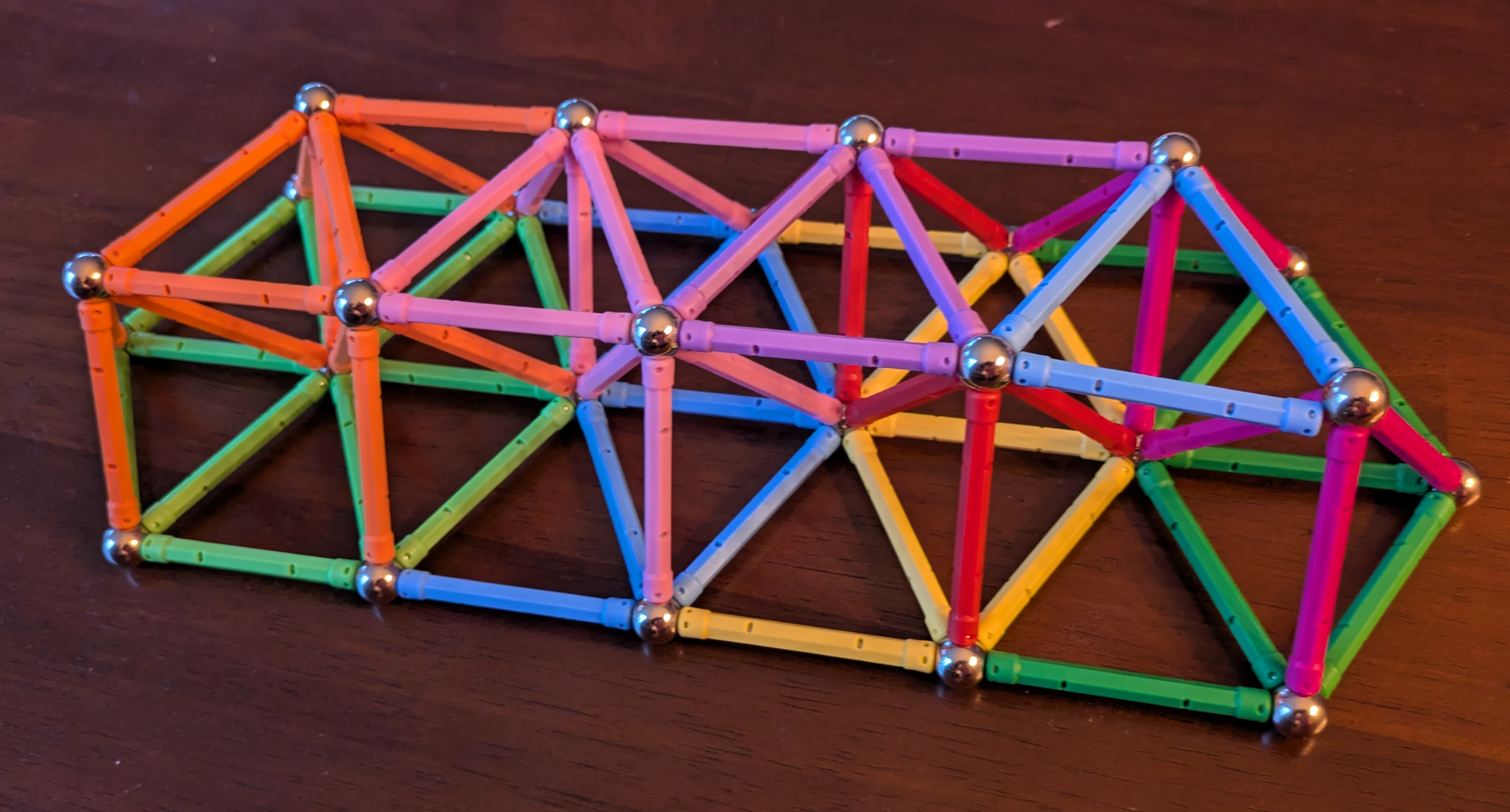}
     \caption{A $3$-arm, another rigid marble graph.}
        \label{fig:karm}
    \end{figure}   
    The arm 
    allows us to enforce distances between points.

    Consider the \defn{spinner} in Figure~\ref{fig:spinner}.
    Like the UFO, the spinner starts with a flat hexagon (in the middle) on which we erect three tetrahedra on each side (on the same three base triangles, though that does not matter), and a final pair of tetrahedra at the top and at the bottom. The top and bottom vertices of the spinner (the apexes of the final two tetrahedra) have distance $4\sqrt{6}/3$ from each other.

\begin{figure}[htbp]
    \centering
    \includegraphics[height=2in]{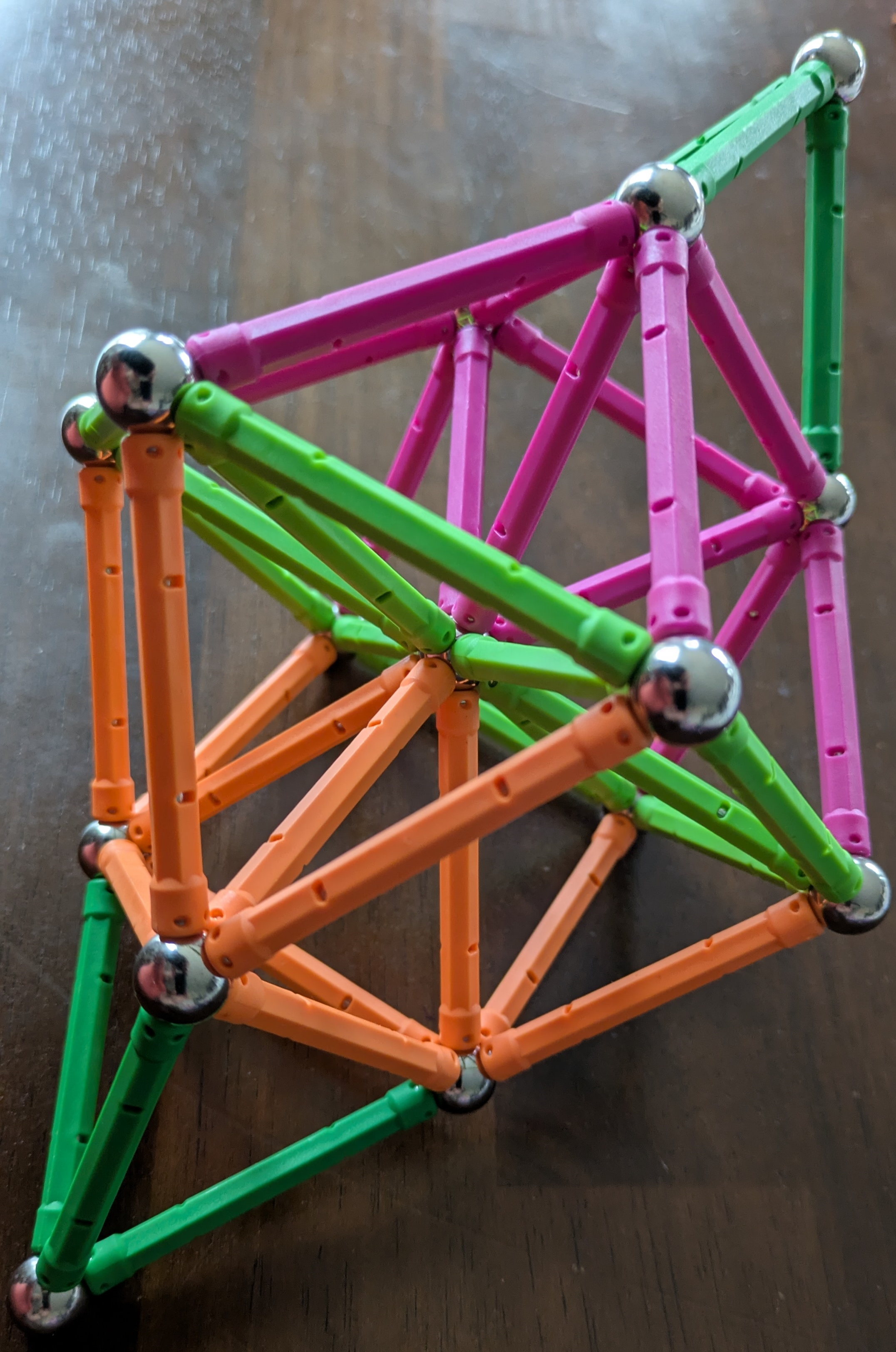}
    \caption{A spinner: yet another rigid marble graph.}
    \label{fig:spinner}
\end{figure}

    It is easy to verify that all the gadgets we have seen so far: $k$-arms, UFOs and spinners are marble graphs, indeed the visualizations we included show how to realize each of the gadgets. 
    We prove in Appendix~\ref{app:riggad} that the gadgets are rigid:

    \begin{claim}\label{claim:rigidgadgets}
     The UFO, the $k$-arm (for every $k \geq 1$) and the spinners are rigid in $\RN^3$.
    \end{claim}

    We now take nine copies of the spinner, and stack them on top of each other (say along the $z$-axis) so that the top vertex of one spinner is identified with the bottom vertex of the next spinner.
    To keep the spinners vertical
    this arrangement we create a rigid column
    called a \defn{mast}, that alternates hexagons and triangles along an axis parallel to the $z$-axis through the spinners;
    we can ensure that a hexagon in the stack of spinners lies in the same plane and with the same orientation as a hexagon in the 
    mast.
    We then connect the top and the bottom spinners to the 
    mast
    using $11$-arms, see Figure~\ref{fig:vgag}. This will be our \defn{vertex gadget}. Note that each vertex gadget contains seven flexible spinners that can rotate.
\begin{figure}[htbp]
    \centering
    \includegraphics[height=2in]{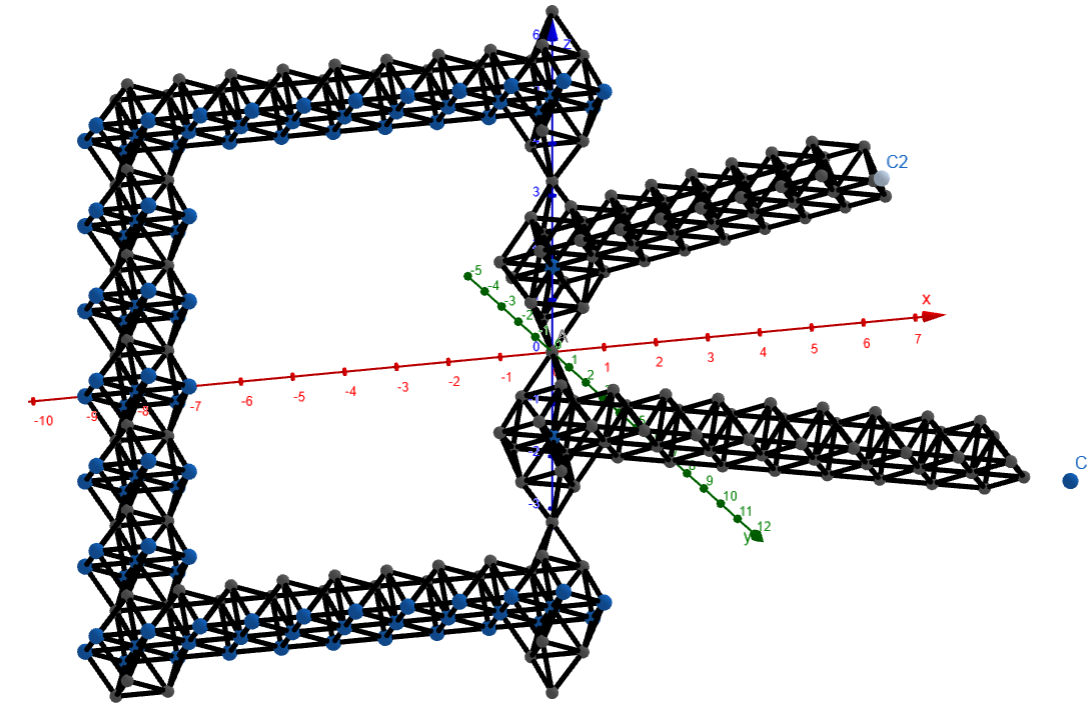}
    \caption{The vertex gadget illustrated for four spinners (instead of $9$), the outer two fixed, the inner two flexible in the sense that they can rotate around the vertex gadget's main axis.
    \Gapp 3d-marble-arms.ggb; this file needs to be opened in ``3D Calculator'' mode, it is slow to render. Points $C1$ and $C2$ can be moved to control the direction of the arms.
    }
    \label{fig:vgag}
\end{figure}    
    
    We can assume that $G$ has max-degree at most $6$, otherwise it is not a penny graph, and
    for the reduction we can let $H$ be a $K_5$ 
    which is not a weak unit distance graph in $\RN^3$.
    Vizing's theorem then implies that the edges of $G$ can be $7$-colored and such a coloring can be found in polynomial time; fix such a coloring.

    We are ready to construct $H$: create a vertex gadget $L(v)$ for each vertex $v \in V(G)$. For every edge $uv \in E(G)$ with
    $c = c(uv)$, connect the top UFOs in the $c$-th flexible spinner of $L(u)$ and the $c$-th flexible spinner in $L(v)$ using a $30$-arm gadget.
    This completes the construction of $H$.

    If $H$ is realizable, the $(x,y)$ coordinates of the $z$-axis passing through $L(v)$ give a location for each $v \in V(G)$ in the plane such that $d(u,v) = 30$ if $uv \in E(G)$ and $d(u,v)>2$ otherwise (assuming $u \neq v$); in particular the locations of two different vertices are distinct;  this is sufficient by Theorem~\ref{thm:penniesER}.

    \begin{figure}[htbp]
    \centering
    \includegraphics[height=2in]{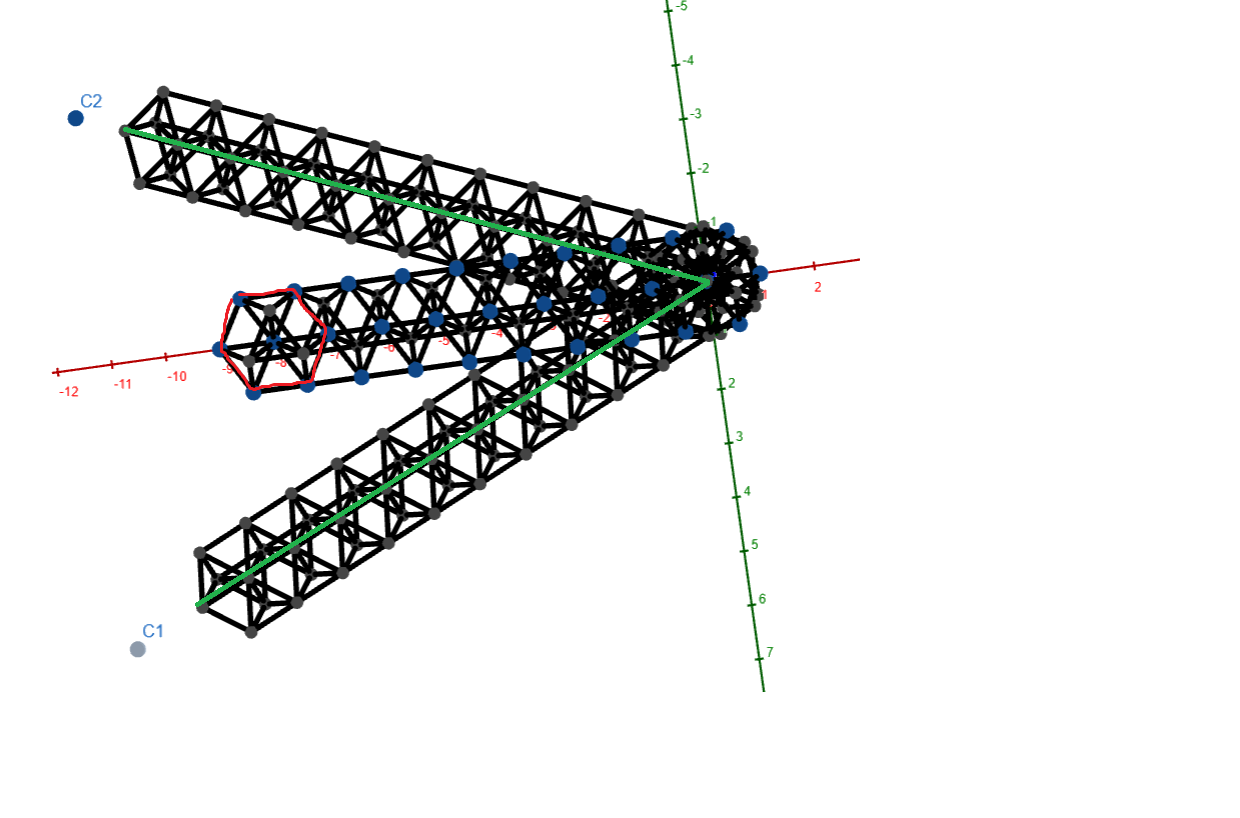}
    \caption{A vertex gadget as seen from above; the two arms form an angle less than $60$ degrees and still have distance greater than $1$ from the vertices of the 
    mast
    in red).}
    \label{fig:vgagtop}
\end{figure} 
    
    If $G$ is a penny graph, we can fix a realization of $G$ in the $xy$-plane, and erect the three-dimensional gadgets of top of this realization as shown in the figures. Then $H$ is a marble graph, using the realizations of the gadgets shown in the figures. There is only one constraint we need to be careful with: 
  the mast of a vertex gadget must be positioned where it does not intersect any of the arms.  In a realization of a penny graph, two edges incident on the same vertex form an angle greater than $60$ degrees.  As shown in the top view of the vertex gadget in Figure~\ref{fig:vgagtop}, because the mast is joined to the central axis of the gadget by an 11-arm, there is room to  place it in the $\ge 60^\circ$ gap between any two consecutive arms.
 Finally, because edges of the penny graph become $30$-arms, no two masts interfere with each other.  
\end{proof}

\section{Rigidity}\label{sec:R}

Two realizations (of a linkage or penny graph) are \defn{congruent} if  the pairwise point distances in each realization are the same,
or, equivalently, if there is an \defn{isometry} (a distance-preserving map) of the plane that maps one realization to the other. A realization is \defn{rigid} if there is an $\varepsilon > 0$ such that any configuration in which each vertex is perturbed by a distance of at most $\varepsilon$ is congruent to the original configuration. In the \defn{rigidity} problem we are given a linkage or penny graph together with a realization of the linkage (or penny graph) and ask whether the realization is rigid. Since realizations of penny graphs require non-rational coordinates we work 
over the set of constructible numbers; the \defn{constructible numbers} are those numbers that can be obtained from $0$ and $1$ using addition, multiplication, inverses, and square roots; constructible numbers can be expressed by a radical expression. Geometrically, a real number is constructible if a line segment with its length can be constructed using compass and straightedge. 

Abbott~\cite{A08} showed that linkage rigidity is \coNP-hard by a reduction from the isolated zero problem \ISO, which we will introduce below; Schaefer~\cite{S13} showed that \ISO\ is complete for \VR, implying that linkage rigidity is \VR-complete. Abel, Demaine, Demaine, Eisenstat, Lynch, and Schardl~\cite{ADDELS25} strengthened this result to show that linkage rigidity remains \VR-complete for non-crossing unit linkages. In this section we take this result one step further by showing that rigidity of penny graphs is \VR-complete.

\begin{theorem}\label{thm:pennyrigVR}
 Testing whether a given realization of a penny graph is rigid is \VR-complete.
\end{theorem}

The proof follows the same pattern as the rigidity hardness proofs in~\cite{A08,S13}. For a family $f = (f_i)_{i \in [n]}: \RN^d \rightarrow \RN^n$ of (explicitly given) polynomials, let $V(f) = \{x \in \RN^d: f(x) = 0\}$.

We start with the problem \ISO: given a family $f = (f_i)_{i \in [n]}: \RN^d \rightarrow \RN^n$ of polynomials such that $0 \in V(f)$, is $0$ an \defn{isolated zero}? That is, is there an $\varepsilon > 0$ such that for every $x \in \RN^d$ with $0 < \lvert x \rvert < \varepsilon$ we have $x \not\in V(f)$? It was shown in~\cite{S13} that this problem is \VR-complete. To apply a result by Abrahamsen and Miltzow~\cite{AM19} later, we need $V(f)$ to be compact; as the following lemma shows, we can assume this to be the case. The proof of the lemma can be found in Appendix~\ref{app:ISOcompact}.

\begin{lemma}\label{lem:ISOcompact}
    Testing whether $0$ is an isolated zero of a family of polynomials $f: \RN^d\rightarrow \RN^n$ is \VR-complete even if $V(f)$ is a compact set.
\end{lemma}

Given a formula $\varphi(x)$ with free variables $x \in \RN^n$, we define
$V(\varphi)$, the  \defn{realization-space} of $\varphi$ as:
\[V(\varphi) := \{x \in \RN^n: \varphi(x)\}.\]
With this we have $V(f) = V(f = 0)$.


\begin{proof}[Proof of Theorem~\ref{thm:pennyrigVR}]
Suppose we are given a family $f: \RN^d\rightarrow \RN^n$ of polynomials such that $0 \in V(f)$ and $V(f)$ is compact. By Lemma~\ref{lem:ISOcompact} it is \VR-complete to tell whether $0$ is isolated in $V(f)$. Using~\cite[Theorem 1]{AM19} we can efficiently compute an instance $\Phi = (\exists x \in [\frac{1}{2},2]^{n'})\ \varphi(x)$ of \ETRINV, with $n' \in \NN$, and a rational point $p \in [\frac{1}{2},2]^{n'}$ such that $0$ is an isolated $0$ in $V(f)$ if and only if $p$ is an isolated point in $V(\varphi)$. 

Now using our reduction in Theorem~\ref{thm:penniesER} we can build a graph $G$ such that every point in $V(\varphi)$ corresponds to a realization of $G$ as a penny graph. It follows that the realization corresponding to $p$ is rigid if and only if $p$ is isolated in $V(\varphi)$. 

It remains to show that the realization corresponding to $p$ only requires constructible coordinates. But this follows from the nature of our gadgets: all points in our linkages can be obtained by specifying their distance from two other points (and possibly selecting one of two possible solutions), which corresponds to intersecting two circles (leading to square root expressions); for example, see how $C$ is obtained from $H$ and $A$ in Figure~\ref{fig:A-frame-proof}. Similarly, the locations of all pennies in our penny gadgets can be described with respect to the (constructible by induction) endpoints of the bars they belong to.
\end{proof}

\defn{Global rigidity} is a stronger rigidity notion that requires all realizations of a linkage (or penny graph) to be congruent. Abel, Demaine, Demaine, Eisenstat, Lynch, and Schardl~\cite{ADDELS25} showed that testing global rigidity of non-crossing matchstick graphs is \VR-complete. Our construction does not immediately yield this result; for example, Hart's A-frame is not globally rigid, since it can bend in one of two ways. It is likely that our linkage and penny graph gadgets can be made globally rigid, but we have not attempted to do so. 

For a recent paper discussing rigidity questions on penny (and marble) graphs, see Dewar et al.~\cite{DGKMN25}.



\section{Grid Embeddings}\label{sec:GE}

At a first glance asking for grid embeddings of penny graphs seems hopeless: an equilateral triangle requires at least one irrational coordinate in any embedding, so no graph containing a $3$-cycle is a penny graph with its vertices lying on a square grid. A second glance suggests two follow-up questions: $(1)$ what can we say about penny graphs on triangular grids, and $(2)$ what if we consider graphs without any $3$-cycles?

We do not have any answers to question $(1)$. It is conceivable that our reduction for proving Theorem~\ref{thm:penniesER} satisfies (or can be modified) to satisfy the property that the \ETRINV\ formula $\varphi$ has a solution in the rationals if and only if $G$ is a penny graph embeddable on a triangular grid. In this case, the triangular grid-embedding variant of penny graphs would be complete for \EQ, the existential theory of the rational numbers, which is not known to be decidable. Graph drawing problems for which it is known that their grid variants are \EQ-complete include the planar slope number~\cite{H17}, recognition of visibility graphs~\cite{CH17}, and right-angle crossing graphs~\cite{S23}. On the other hand, it is also conceivable that $G$ 
can be constructed from $\varphi$ in such a way that there always is a triangular grid embedding of the penny graph if $\varphi$ is satisfiable. In that case the problem would be \ER-complete.

Question $(2)$ can be sharpened: How hard is it to recognize whether $G$ is a penny graph if $G$ is bipartite or a tree? As we mentioned earlier, the problem is \NP-hard for trees 
in the fixed embedding variant~\cite{BDLRST15}. This suggests the question of whether a tree penny graph can always be realized on a grid, and if so, how small a grid? If we could show that the grid has exponential---or even polynomial---size, then recognizing tree penny graphs (with a given embedding) would be \NP-complete. Unfortunately, we seem to be very far from showing such a result, but we can show that a doubly-exponential size grid is sufficient for trees.

\begin{theorem}\label{thm:treegrid}
 A tree which is a penny graph can be embedded on a grid of size at most $2^{2^{O(n^k)}}$ where $n$ is the order of the tree.
\end{theorem}

The same result also holds for matchstick graphs with essentially the same proof.

\begin{proof}
The main trick in the proof is to parameterize the embeddings of $T$ in such a way that rational parameters correspond to a grid embedding. To do this we view $T$ as a rooted tree; to each directed edge $uv$ in the tree we associate a real number $m_{uv} \in \RN$. The collection of these numbers describes a potential penny graph drawing of $T$ as follows: let the root of $T$ be realized by a unit disk centered at the origin. Assuming the penny $D_u$ for $u$ has already been placed, we place the penny $D_v$ for $v$ such that the slope between the 
bottommost point of $D_u$ and the topmost point of $D_v$ has slope $m_{uv}$ and $D_u$ and $D_v$ touch; in an abuse of notation we write $u$ and $v$ for the centers of $D_u$ and $D_v$, respectively. See Figure~\ref{fig:treetwopennies}. We call this drawing {\em potential}, since disks belonging to non-adjacent vertices may overlap, which is not legal in a penny drawing, of course.
\begin{figure}[htbp]
    \centering
    \includegraphics[height=2in]{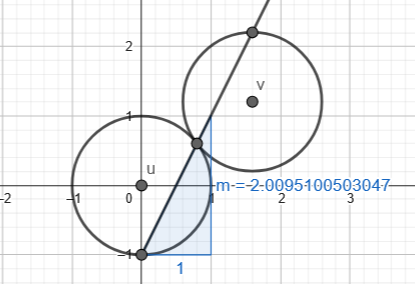}
    \caption{Two pennies $D_u$ and $D_v$ in a tree with parameter (slope) $m = m_{uv}$.}
    \label{fig:treetwopennies}
\end{figure}   
This parametrization is based on the parametrization of Pythagorean triangles. Assuming that $u$ has rational coordinates, we have the following property: $m_{uv}$ is rational if and only if $v$ has rational coordinates. This is because the contact point between the two disks has that property and $v$ is the reflection of $u$ in that point. We conclude that $(m_{uv})_{uv \in E(T)}$ describes a placement of pennies (which may not be legal) such that the centers of all pennies are rational if and only if all the $m_{uv}$ are rational. If there is a penny graph realization of $T$ on a grid, we can move the root of the tree to the origin, and (rationally) rotate the embedding such that there are no two center vertices with the same $x$-coordinate. This realization can then be described using rational parameters $(m_{uv})_{uv \in E(T)}$.

Let $S$ be the semialgebraic set consisting of all $(m_{uv})_{uv \in E(T)}$ which correspond to a legal penny graph drawing so any pair of disks not corresponding to an edge in $T$ have distance greater than $0$. It follows that $S$ is an open set. If $T$ is a penny graph, then $S$ is not empty, and, by a result of Kachiyan and Porkolab~\cite[Corollary 2.6]{KP97}, contains a rational point of at most exponential bit-complexity, which can be turned into an embedding of $T$ on a grid of size at most $2^{2^{O(n^k)}}$.
\end{proof}

The bound in Theorem~\ref{thm:treegrid} is unlikely to be optimal, but we do not see any simple tools that would remove even one of the exponentiations, which would imply an \NP\ upper bound on the penny graph recognition problem for trees. Can one construct trees  which require a single-exponential size grid in a penny graph embedding?

We do not expect that recognizing whether a tree is a penny graph is \ER-complete (though we have no supporting evidence), but what about bipartite graphs? The problem with bipartite graphs or any graph not containing $3$-cycles is that they do not seem to be rigid
which makes encoding hard. Our reduction certainly relies on the presence of $3$-cycles. So we again face the combinatorial question of whether bipartite penny graphs can always be embedded on a grid. The complexity of the question is also wide 
open.

\section{Open Questions}

We saw that penny graph recognition is \ER-complete; our construction makes essential use of $3$-cycles, to rigidify the realizations, and induced $4$-cycles, to simulate a degree-$3$ vertex with variable angles. 
\begin{question}
    Does penny graph recognition remain \ER-hard if the penny graph does not contain $3$-cycles and/or induced $4$-cycles?
\end{question} 
(Eppstein~\cite{E18} showed that penny graphs without $3$-cycles are $2$-degenerate.) Encoding without these tools appears difficult; on the other hand, it is not clear that forbidding these substructures makes the problem easier to solve. We already discussed the tree case in Section~\ref{sec:GE}.  

Another restriction to consider is the maximum degree of the penny graph; the gadgets constructed in our reduction have maximum degree $6$; with a bit more care we can avoid degree-$6$ vertices so that penny graph recognition remains \ER-hard for maximum degree $5$. Since all graphs of maximum degree $2$ are penny graphs, this leaves open the cases of maximum degree $3$ and $4$. 
\begin{question}
    Does penny graph recognition remain \ER-hard for graphs of maximum degree~$3$ or $4$?
\end{question}
We conjecture \ER-hardness for maximum degree~$4$. The same question can also  be asked for matchstick graphs.

Breu and Kirkpatrick~\cite{BK96} introduced the approximate penny graph recognition problem, where the disks have radii (not necessarily the same) in an interval $[1,\rho]$; they showed that this problem is \NP-hard. Using standard methods---see the case of approximate \RAC-drawings, for example~\cite[Corollary 18]{S23}---it can be shown that the problem remains \ER-hard for $\rho = 2^{-2^{n^c}}$ for some constant $c>0$, for $n$-vertex graphs, but the complexity of the approximate penny graph recognition problem in which $\rho$ itself is constant is open.

\begin{question}
    Does penny graph recognition remain \ER-hard if the radii of the pennies can range in a fixed interval $[1,\rho]$?
\end{question}

Very little is known about how hard it is to recognize ball contact graphs---with or without unit radius---in higher dimensions, see Remark~\ref{rem:BCGhigher}.
\begin{question}
    What is the complexity of recognizing [unit] ball contact graphs in more than three dimensions?
\end{question}
There is no complexity lower bound even for $\RN^5$.

Penny graphs have been studied on surfaces other than the plane, e.g., Lorand showed that $K_5$ and $K_{3,3}$ are penny graphs on the flat torus~\cite{L24}. It is easy to see that our constructions can be extended to show that whether a graph is a penny graph on a flat torus is still \ER-complete (by constructing a gadget that blocks access to the topology). Less obvious is what happens on other surfaces, like the sphere, or hyperbolic surfaces. Hyperbolic unit disk intersection graphs are known to be \ER-complete~\cite{BBDJ23}
to recognize.
\begin{question}
    Is recognizing penny graphs on other surfaces, such as the sphere or the torus, \ER-complete?
\end{question}


\bibliographystyle{plainurl}
\bibliography{penny}

\clearpage
\appendix
\section{Hart's A-frame}\label{app:AFrame}

We prove Claim~\ref{claim:A-frame}.

\begin{figure}[h!]
    \centering    \includegraphics[width=0.45\linewidth]{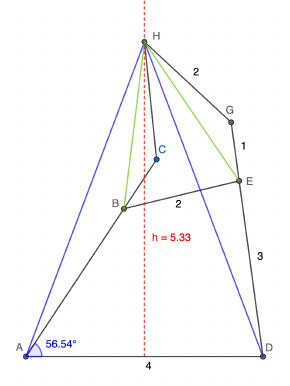}
    \ \ \ \includegraphics[width=0.35\linewidth]{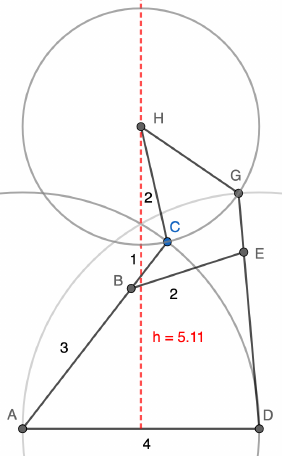}
    \caption{Hart's A-frame.}
    \label{fig:A-frame-proof}
\end{figure}

\begin{claim}
\label{claim:A-frame-proof}
The apex of Hart's A-frame moves along a straight line   that is a perpendicular bisector of the base.  
\end{claim}
\begin{proof}
We use notation as in Figure~\ref{fig:A-frame-proof}, and we use, e.g., $HA$ to mean both the bar and its length. 

(1) Triangles $HCA$ and $BCH$ are similar because the angle at $C$ is the same and $HC/CA = BC/CH = 1/2$.
Thus $HB/HA = 1/2$.
Similarly, $HE/HD = 1/2$.
This implies that $HB/HE = HA/HD$, a property we will use in step (3).

(2) Triangles $BHE$ and $AHD$ are similar because all three corresponding pairs of edges have ratio $\frac{1}{2}$.  Thus angle $BHE$ equals angle $AHD$.  Subtracting the common angle $BHD$ we get angle $AHB$ equals angle $DHE$.

(3) Triangles $AHB$ and $DHE$ are similar because the angle at $H$ is the same and $HB/HE = HA/HD$ (from (1)).  Since $AB = DE$, the triangles are in fact congruent.  Thus $AH = DH$.  So $H$ is the apex of an isosceles triangle on base $AD$, and $H$ moves on the perpendicular to the base.


We add one further point for use later:

(4) Since triangles $BHE$ and $AHD$ are similar, and $AHD$ is isoceles, so is $BHE$.  Thus $BH=EH$ and triangles $BHC$ and $EHG$ are congruent.  This implies that angle $CHG$ equals $AHD$.
\end{proof}

\begin{claim}
For each value of $h \in (4,4 \sqrt{2})$ there are exactly two configurations of Hart's A-frame (left-leaning and right-leaning) that put $H$ at height $h$. Both configurations are non-crossing. Furthermore, 
$h$ 
and all the angles of the A-frame are  continuous monotone functions of the angle~$CAD$.
\end{claim}
\begin{proof} 
We show that any point $H$ on the perpendicular bisector of $AD$ at height $h \in (4,4 \sqrt{2})$ yields two valid configurations of the A-frame.  We basically follow the above proof backwards. 
For each position of $H$ on the perpendicular bisector of $AD$, point $C$ is at the intersection of two circles as shown in Figure~\ref{fig:A-frame-proof}(\emph{right}). 
Similarly, $G$ is at the intersection of two circles. 
Edge 
$BE$ has length 2 iff we make the same choice (left-leaning or right-leaning) for $C$ and $G$.  

Clearly, $h$ is a continuous monotone function of angle $CAD$. As for the angles, the only non-obvious one is angle $CHG$: by point (4) in the proof of Claim~\ref{claim:A-frame-proof}, angle $CHG$ is equal to $AHD$, which changes monotonically.  
\end{proof}

\section{Proof of Lemma~\ref{lem:ISOcompact}}\label{app:ISOcompact}

Our goal is to show that testing whether $0$ is an isolated zero of $g$ is \VR-complete even if $V(g)$ is a compact set. We closely follow the proof of~\cite[Corollary 3.10]{S13}.

We make use of the result that testing whether a family of polynomials $f = (f_i)_{i \in [n]}: \RN^d \rightarrow \RN^n$ of total degree at most $2$ has a zero in the \defn{($d$-dimensional) unit ball} $B^d(0,1) := \{x \in \RN^d: \lvert x \rvert \leq 1\}$, that is $V(f) \cap B^d(0,1) \neq \emptyset$, is \ER-complete, by~\cite[Lemma 3.9]{S13}. We define a new family of polynomials as follows:
\[g_i(x_1, \ldots,x_d,y_1,y_2,y_3) = y_1^2f_i(x_1/y_1,\ldots,x_d/y_1),\]
where $i \in [n]$, and
\begin{align*}
g_{d+1}(x_1, \ldots,x_d,y_1,y_2,y_3) & = y_2^2-y_1^2+\sum_{i \in [d]} x_i^2 ,\\
g_{d+2}(x_1, \ldots,x_d,y_1,y_2,y_3) & = y_1^2+(1-y_3)^2 - 1.
\end{align*}
Then $g = (g_i)_{i \in [d+2]}$ is a family of polynomials, the first $d$ of them homogeneous, such that $(x_1, \ldots, x_n, y_1,y_2,y_3) \in V(g)$ if and only if $y_1 = 0$ or $y_1 \neq 0$ and $(x_1/y_1,\ldots,x_d/y_1) \in V(f)$, and 
$\sum_{i \in [d]} x_i^2 + y_2^2-y_1^2 = 0$ as well as
$y_1^2+(1-y_3)^2 = 1$.
By construction, $0$ is a zero of $g$. Also, if $(x_1,\ldots, x_d, y_1,y_2,y_3) \in V(g)$, then the condition on $g_{d+2}$ implies that $(y_1,y_3) \in B_2(0,1)$, which, using that $g_{d+1}$ is zero, then implies that $y_2 \in B_1(0,1)$ and $(x_1,\ldots, x_d) \in B_d(0,1)$, so $V(g) \subseteq B_{d+3}(0,1)$ is bounded and therefore, by the continuity of $g$, compact. 

\begin{claim}
$f$ has a zero in $B^d(0,1)$ if and only if $0$ is not an isolated zero of $g$.
\end{claim}
\begin{proof}
Suppose that $(x_1, \ldots, x_d) \in B^d(0,1)\cap V(f)$. Then $s := \sum_{i \in [d]} x_i^2 \leq 1$. For arbitrary $y_1 \in (0,1)$ 
define $x'_i = y_1x_i$, $i \in [d]$, $y_2 = y_1 \sqrt{1-s}$, and $y_3 = 1-\sqrt{1-y_1^2}$. 
Then
$(x'_1, \ldots x'_d,y_1,y_2,y_3) \in V(G)$
since $g_i(x'_1, \ldots x'_d,y_1,y_2,y_3) = y_1^2f(x_1,\ldots,x_d) = 0$ for $i \in [d]$, and 
\begin{align*}
   g_{d+1}(x'_1, \ldots x'_d,y_1,y_2,y_3) & = y_2^2-y_1^2+\sum_{i \in [d]} (x'_i)^2  \\
            &= y_2^2-y_1^2+y_1^2 \sum_{i \in [d]} x_i^2 \\
            &= y_2^2-y_1^2+y_1^2 s  \\
            & = 0
\end{align*}
as well as 
\[g_{d+2}(x'_1, \ldots x'_d,y_1,y_2,y_3) =y_1^2 + (1-y_3)^2 -1 = 0.\]
So for every $y_1 \in (0,1)$ we have found a zero $(x'_1, \ldots x'_d,y_1,y_2,y_3)$ in $V(g)$ and since, by definition, $x'_i$, $y_2$ and $y_3$ converge to $0$ if $y_1$ does, $0$ is not an isolated zero of $g$.

For the other direction, let us assume that $0$ is not isolated in $V(g)$. So for every $k > 0$ there is a point $p_k = (x_{k,1},\ldots, x_{k,d},y_{k,1}, y_{k,2}, y_{k,3}) \in B^{d+3}(0,1/k) -\{0\}$ such that $f(p_k) = 0$. We will argue that the first $d$ coordinates $(x_{k,1},\ldots, x_{k,d})$ of $p_k$ converge to a point $x \in B^d(0,1)$. 

If  $y_{k,1} = 0$, then $g_{d+1}(p_k) = 0$ implies that $y_{k,2} = \sum_{i \in [d]} x_i^2 = 0$, and $g_{d+2}(p_k) = 0$ forces $y_{k,3} = 0$. So $p_k = 0$ which contradicts our choice of $p_k$. We conclude that $y_{k,1} \neq 0$, allowing
us to define $x'_k := (x_{k,1}/y_{k,1}, \ldots, x_{k,d}/y_{k,1})$.
Then
\[f_i(x'_k) = 0\]
for $i \in [d]$. Since $g_{d+1}(p_k) = 0$ we know that
$\sum_{i \in [d]}x_{k,i}^2 \leq y_{k,1}^2$, so $\sum_{i \in [d]}(x_{k,i}/y_{k,1})^2 \leq 1$ and $x'_k = (x_{k,1}/y_{k,1}, \ldots, x_{k,d}/y_{k,1}) \in B_d(0,1)$. We conclude that $(x'_k)$ is a sequence of zeros of $f$ in $B_d(0,1)$ and, since $B_d(0,1)$ is compact the sequence has a limit point $x \in B_d(0,1)$ which belongs to $V(f)$ by continuity.
\end{proof}

\section{Proof of Claim~\ref{claim:rigidgadgets}}\label{app:riggad}

 First of all, we observe that the regular tetrahedron is rigid (up to an isometry, possibly orientation-reversing). Similarly, the \defn{regular triangular bipyramid}, the union of two regular tetrahedra along a face is rigid, since the two apexes have to be distinct. 

 To see that the UFO gadget is rigid, we view it from the top, as seen in the left half of Figure~\ref{fig:smallUFOrigid}. We can then consider the plane angles (as in the projection from the top) between consecutive spokes of the wheel. Each of these angles must be at most $2\pi/6$: If one of the angles, say between $uv$ and $uw$ as seen from above were larger than $2\pi/6$, the actual angle between $uv$ and $uw$ would be larger than $2\pi/6$ implying that $vw$ cannot have the same length as $uv$ and $uw$ (see the right half of the figure). Since the plane angles must sum up to $2\pi$ we conclude that the angles are all exactly $2\pi/6$ which implies that the hexagon is flat, forcing the embedding to be rigid. 

\begin{figure}
    \centering    \includegraphics[height=2in]{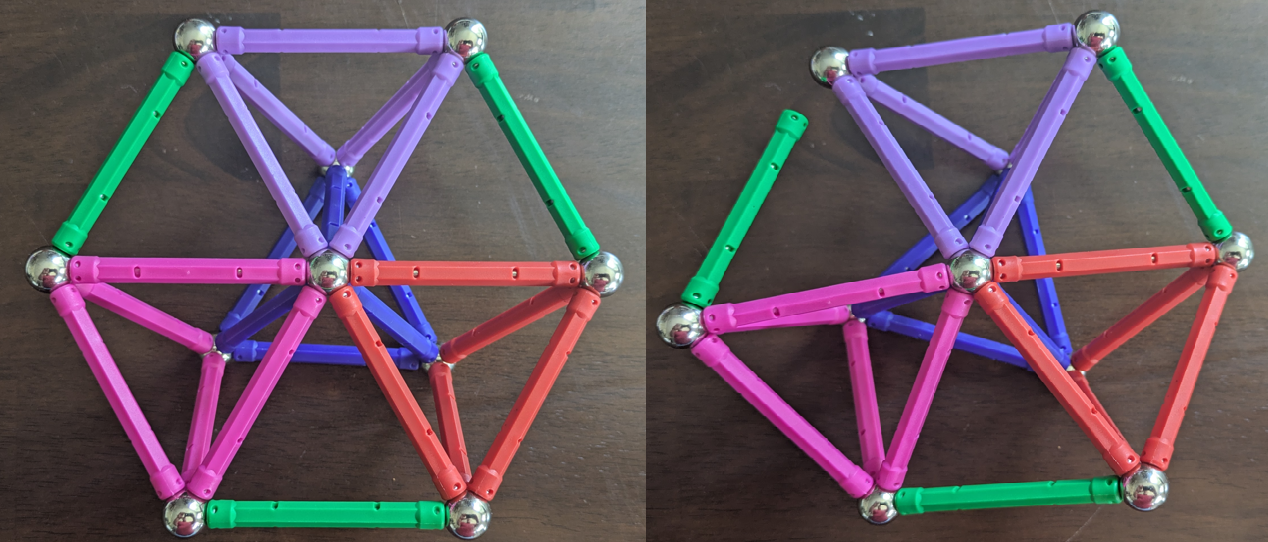}
    \caption{Rigidity of the UFO.}
    \label{fig:smallUFOrigid}
\end{figure}

 For the $k$-arm we work with the following observation: If two rigid graphs share a tetrahedron, then their union is rigid. Since we can think of a $k$-arm as the union of $k$ UFO gadgets, every two of which have a tetrahedron in common, the $k$-arms are rigid.

 The spinner is the union of four rigid graphs, two UFOs and two tetrahedra; to this union we can add three triangular bipyramids, two overlapping a UFO and the tetrahedron attached to it, and one overlapping the two UFOs; the union of these seven graphs is still the spinner, and the bipyraminds share a tetrahedron with each of the graphs they overlap, implying that the whole spinner is rigid. 

\end{document}